\documentclass[a4paper,UKenglish,cleveref, autoref, thm-restate]{lipics-v2021}

\usepackage{lineno}
\nolinenumbers
\usepackage{graphicx}
\usepackage{amsfonts} 
\usepackage{amsmath}
\usepackage{amssymb}
\usepackage{algorithm}
\usepackage{algpseudocode}
\usepackage[framemethod=tikz]{mdframed}
\usepackage{stmaryrd}
\usepackage{cite}
\usepackage{colonequals}
\usepackage{subcaption}
\usepackage{pgfplots}
\usepgfplotslibrary{external}
\usepackage{algorithm}
\usepackage{algpseudocode}
\algtext*{EndWhile}
\algtext*{EndIf}
\algtext*{EndFor}
\usepackage{amsmath}

\usepackage{tikz}
\usetikzlibrary{calc}
\usetikzlibrary {arrows.meta,automata,positioning}
\usetikzlibrary {shapes.geometric}
\usepackage{tikz-cd}
\usepackage[all]{xy} 
\tikzstyle{point}=[circle,  inner sep=2pt, fill]

\algdef{SE}[DOWHILE]{Do}{doWhile}{\algorithmicdo}[1]{\algorithmicwhile\ #1}%
\algnewcommand\algorithmicforeach{\textbf{for each}}
\algdef{S}[FOR]{ForEach}[1]{\algorithmicforeach\ #1\ \algorithmicdo}

\newenvironment{proofs}{%
  \renewcommand{\proofname}{Proof Sketch}\proof}{\endproof}

\newcommand{\littletaller}{\mathchoice{\vphantom{\big|}}{}{}{}}

\newcommand\restr[2]{{
  \left.\kern-\nulldelimiterspace 
  #1 
  \littletaller 
  \right|_{#2} 
  }}

\newcommand{\setPriorities}{\mathbb{P}}
\newcommand{\nat}{\mathbb{N}}
\newcommand{\eve}{\exists}
\newcommand{\adam}{\forall}
\newcommand{\indPlay}[3]{p_{#1, #2;#3}}
\newcommand{\pg}[1]{\mathcal{#1}}
\newcommand{\opg}[1]{\mathcal{#1}}
\newcommand{\shortcut}[1]{\mathcal{C}(#1)}
\newcommand{\interface}{\mathsf{IO}}
\newcommand{\defeq}{:=}
\newcommand{\dr}{\mathbf{r}}
\newcommand{\dl}{\mathbf{l}}
\newcommand{\en}{I}
\newcommand{\ex}{O}
\newcommand{\enr}{I_{\dr}}
\newcommand{\enrarg}[1]{i_{\dr, #1}}
\newcommand{\enl}{I_{\dl}}
\newcommand{\enlarg}[1]{i_{\dl, #1}}
\newcommand{\exr}{O_{\dr}}
\newcommand{\exrarg}[1]{o_{\dr, #1}}
\newcommand{\exl}{O_{\dl}}
\newcommand{\exlarg}[1]{o_{\dl, #1}}
\newcommand{\seqcomp}{\fatsemi}

\newcommand{\sd}[1]{\mathbb{#1}}
\newcommand{\semantics}[1]{\llbracket #1 \rrbracket}
\newcommand{\semPlay}[1]{
\semEntrance{#1}
}
\newcommand{\semOPGs}[1]{
\semEntrance{#1}
}

\newcommand{\posStrategies}[2]{\Sigma^{#1}_{#2,\mathrm{p}}}

\newcommand{\priorityOrder}{\preceq_{\setPriorities}}

\newcommand{\priorityStOrder}{\prec_{\setPriorities}}
\newcommand{\enrichPriorityOrder}{\preceq_{\overline{\setPriorities}}}
\newcommand{\domainOrder}{\preceq_{\mathbf{D}}}

\newcommand{\playOrder}{\preceq_{\mathbf{P}}}
\newcommand{\upperOrder}{\preceq_{\mathbf{U}}}
\newcommand{\lowerOrder}{\preceq_{\mathbf{L}}}
\newcommand{\succLowerOrder}{\succeq_{\mathbf{L}}}

\newcommand{\domain}[1]{\mathcal{D}_{#1}}
\newcommand{\paretoCurve}[2]{\mathcal{S}(#1, #2)}
\newcommand{\paretoCurves}[1]{\mathcal{S}(#1)}
\newcommand{\edgeConstructor}[3]{\mathbb{G}(#1, #2, #3)}
\newcommand{\semEntrance}[1]{
\llparenthesis #1 \rrparenthesis
}
\newcommand{\potentialFuncs}[1]{\mathbb{Q}_{#1}}
\newcommand{\potentialOrder}[1]{\preceq_{\mathbb{Q}_{#1}}}

\newcommand{\dualPriority}[1]{#1^d}

\newcommand{\queryAnswer}[1]{r(#1)}
\newcommand{\setPoints}[1]{\mathbf{R}_{#1}}
\newcommand{\resultOrder}[1]{\preceq_{\mathbf{R}_{#1}}}
\newcommand{\succPointOrder}[1]{\succeq_{\mathbf{R}_{#1}}}
\newcommand{\result}{\mathbf{r}}

\newcommand{\myparagraph}[1]{\smallskip\noindent \emph{#1}}

\newcommand{\worst}[1]{#1^{\mathrm{min}}}
\newcommand{\best}[1]{#1^{\mathrm{max}}}

\crefname{definition}{Def.}{Defs}
\crefname{theorem}{Thm.}{Thms}
\crefname{proposition}{Prop.}{Props}
\crefname{remark}{Rem.}{Remarks}
\crefname{lemma}{Lem.}{Lemmas}
\crefname{proof}{Proof.}{Proofs}
\crefname{appendix}{Appendix}{Appendixes}
\crefname{algorithm}{Alg.}{Algs}
\crefformat{section}{{\S}#2#1#3}
\crefname{figure}{Fig.}{Figs}
\Crefname{equation}{}{}

\title{Pareto Fronts for Compositionally Solving String Diagrams of Parity Games} 

\author{Kazuki Watanabe}{National Institute of Informatics, Tokyo\\ The Graduate University for Advanced Studies (SOKENDAI), Hayama}{kazukiwatanabe@nii.ac.jp}{https://orcid.org/0000-0002-4167-3370}{}
\authorrunning{Kazuki Watanabe} 

\Copyright{Kazuki Watanabe}



\category{} 

\relatedversion{} 

\funding{We are supported by the JST grant No. JPMJAX23CU.}

\acknowledgements{We would like to thank the anonymous referees for their helpful suggestions, which have enabled us to improve this article.}

\EventEditors{Corina C\^{i}rstea and Alexander Knapp}
\EventNoEds{2}
\EventLongTitle{11th Conference on Algebra and Coalgebra in Computer Science (CALCO 2025)}
\EventShortTitle{CALCO 2025}
\EventAcronym{CALCO}
\EventYear{2025}
\EventDate{June 16--18, 2025}
\EventLocation{University of Strathclyde, UK}
\EventLogo{}
\SeriesVolume{342}
\ArticleNo{14}

\begin{document}

\maketitle

\begin{abstract}
Open parity games are proposed as a compositional extension of parity games with algebraic operations, forming string diagrams of parity games.
A potential application of string diagrams of parity games is to describe a large parity game with a given compositional structure and solve it efficiently as a divide-and-conquer algorithm by exploiting its compositional structure.

Building on our recent progress in \emph{open Markov decision processes}, we introduce \emph{Pareto fronts} of open parity games, offering a framework for multi-objective solutions.
We establish the \emph{positional determinacy} of open parity games with respect to their Pareto fronts through a novel translation method.
Our translation converts an open parity game into a parity game tailored to a given single-objective.
Furthermore, we present a simple algorithm for solving open parity games, derived from this translation that allows the application of existing efficient algorithms for parity games.
Expanding on this foundation, we develop a compositional algorithm for string diagrams of parity games.
\end{abstract}

\section{Introduction}
\emph{Parity games} (on finite graphs) have been actively studied, and the search for efficient algorithms is a central topic for parity games.
A key property of parity games is \emph{positional determinacy}~\cite{DBLP:conf/lop/Emerson85,DBLP:conf/focs/EmersonJ91}, which underpins their decidability
by ensuring that positional strategies suffice to determine the winner at any node. 
Calude et al.~\cite{DBLP:journals/siamcomp/CaludeJKLS22} show that parity games are solvable in quasi-polynomial-time, leading to the development of subsequent quasi-polynomial-time algorithms~\cite{DBLP:conf/lics/JurdzinskiL17,DBLP:journals/sttt/FearnleyJKSSW19,DBLP:conf/lics/Lehtinen18}.
The significance of these advancements extends beyond theoretical interest; they have practical applications in various verification problems.
Notably, many verification problems can be reduced to solving parity games,
including model checking with temporal specifications~\cite{DBLP:journals/tcs/Kozen83,DBLP:conf/stoc/KupfermanV98,DBLP:journals/tcs/EmersonJS01,DBLP:conf/lics/Ong06,DBLP:conf/lics/KobayashiO09}.

\begin{figure}[t]
    \centering
    \begin{minipage}[c]{0.3\hsize} 
        \centering 
        \scalebox{0.6}{
    \begin{tikzpicture}[
innodeEve/.style={draw, circle, minimum size=0.5cm,fill=white},
innodeAdam/.style={draw, diamond, minimum size=0.5cm,fill=white},
interface/.style={draw, rectangle, minimum size=0.5cm},]
\fill[lime] (-0.7cm, -1.2cm)--(-0.7cm, 1.2cm)--(4.9cm, 1.2cm)--(4.9cm, -1.2cm)--cycle;
        \node[interface,fill=white, yshift=0.5cm] (s1) {\scalebox{0.8}{$\enrarg{1}$}};
        \node[interface,fill=white, yshift=-0.5cm] (s9) {\scalebox{0.8}{$\exlarg{1}$}};
        \node[inner sep=0,right=-1.3cm of s9] (exl1) {};
        \node[inner sep=0,right=-1.3cm of s1] (enr1) {};
        \node[innodeEve,right=0.75cm of s1] (s2) {\scalebox{0.8}{$a$}};
        \node[innodeAdam,right=0.75cm of s9] (s3) {\scalebox{0.8}{$b$}};
        \node[innodeAdam,right=0.75cm of s2] (s4) {\scalebox{0.8}{$c$}};
        \node[innodeEve,right=0.75cm of s3] (s5) {\scalebox{0.8}{$d$}};
        \node[interface,fill=white,right=0.75cm of s4] (s7) {\scalebox{0.8}{$\exrarg{1}$}};
        \node[inner sep=0,right=0.6cm of s7] (exr1) {};
        \node[interface,fill=white,right=0.75cm of s5] (s8) {\scalebox{0.8}{$\enlarg{1}$}};
        \node[inner sep=0,right=0.6cm of s8] (enl1) {};
        \draw[->] (enr1) -> (s1);
        \draw[->] (s9) -> (exl1);
        \draw[->] (s1) -> node [above] {\scalebox{0.8}{$0$}} (s2);
        \draw[->] (s2) -> node [above] {\scalebox{0.8}{$1$}} (s9);
        \draw[->] (s3) -> node [above] {\scalebox{0.8}{$1$}} (s9);
        \draw[->] (s2) -> node [above] {\scalebox{0.8}{$1$}} (s4);
        \draw[->] (s4) -> node [left] {\scalebox{0.8}{$0$}} (s5);
        \draw[->] (s5) -> node [above] {\scalebox{0.8}{$2$}} (s3);
        \draw[->] (s3) -> node [left] {\scalebox{0.8}{$1$}} (s2);
        \draw[->] (s4) -> node [above] {\scalebox{0.8}{$1$}} (s7);
        \draw[->] (s8) -> (s5);
        \draw[->] (s7) -> (exr1);
        \draw[->] (enl1) -> (s8);
    \end{tikzpicture}
    }
    \end{minipage}
    \begin{minipage}[c]{0.3\hsize} 
        \centering 
        \includegraphics[scale=0.3]{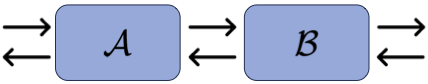}
    \end{minipage}
     \begin{minipage}[c]{0.3\hsize} 
        \centering 
        \includegraphics[scale=0.3]{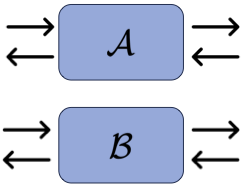}
        \label{subfig:sum} 
    \end{minipage}
    \caption{(i) an open parity game (oPG), (ii) the sequential composition $\opg{A}\seqcomp \opg{B}$ of oPGs $\opg{A},\opg{B}$, and (iii) the sum $\opg{A}\oplus \opg{B}$ of oPGs $\opg{A},\opg{B}$.}
    \label{fig:string_diagrams}
\end{figure}

Compositionality in verification is a property that enables divide-and-conquer algorithms by exploiting compositional structures, aiming to solve the notorious state space explosion problem.
Clarke et al.~\cite{DBLP:conf/lics/ClarkeLM89} provided the first compositional model checking algorithm, and since then, many compositional model checking algorithms have been proposed for LTL~\cite{DBLP:conf/cav/Kaivola92}, Markov decision processes~\cite{DBLP:journals/iandc/KwiatkowskaNPQ13,WVHRJ2024accepted,DBLP:conf/cav/WatanabeEAH23}, and higher-order model checking~\cite{DBLP:conf/csl/TsukadaO14}.
Recently, we introduced \emph{string diagrams of parity games}~\cite{DBLP:journals/corr/abs-2112-14058} as a compositional extension of parity games.
We proposed \emph{open parity games} that can be composed using two algebraic operations: the (bidirectional) sequential composition and the sum.
\cref{fig:string_diagrams} shows an open parity game and illustrates two algebraic operations over open parity games $\opg{A},\opg{B}$ in string diagrams of parity games.  
Although we did not provide a compositional algorithm in the original work~\cite{DBLP:journals/corr/abs-2112-14058},
we developed one for \emph{string diagrams of mean payoff games}~\cite{DBLP:journals/corr/abs-2307-08034}.
The methodology developed is also adaptable to string diagrams of parity games.

A major difficulty for compositional algorithms arises from the \emph{multi-objective perspective}\footnote{This differs from having several parity objectives.} on individual components: 
typically, even if a given ``global'' objective is single-objective, which is the parity condition for parity games, ``local'' objectives on components become multi-objective.

\begin{figure}[t]
        \centering 
        \scalebox{0.9}{
    \begin{tikzpicture}[
innodeEve/.style={draw, circle, minimum size=0.5cm,fill=white},
innodeAdam/.style={draw, diamond, minimum size=0.5cm,fill=white},
interface/.style={draw, rectangle, minimum size=0.5cm},]
        \fill[orange] (-0.7cm, -1.2cm)--(-0.7cm, 1.2cm)--(3.4cm, 1.2cm)--(3.4cm, -1.2cm)--cycle;
        \node[inner sep=0] at (3.2cm, 1.0cm) (A) {$\opg{A}$};
        \node[interface,fill=white] (s1) {\scalebox{0.8}{$\enrarg{1}$}};
        \node[inner sep=0,right=-1.3cm of s1] (enr1) {};
        \node[innodeEve,right=0.75cm of s1] (s2) {\scalebox{0.8}{$a$}};
        \node[interface,fill=white,right=0.75cm of s2, yshift=0.65cm] (s7) {\scalebox{0.8}{$\exrarg{1}$}};
        \node[inner sep=0,right=0.6cm of s7] (exr1) {};
        \node[interface,fill=white,right=0.75cm of s2, yshift=-0.65cm] (s8) {\scalebox{0.8}{$\exrarg{2}$}};
        \node[inner sep=0,right=0.6cm of s8] (exr2) {};
        \draw[->] (enr1) -> (s1);
        \draw[->] (s1) -> node [above] {\scalebox{1}{$0$}} (s2);
        \draw[->] (s2) -> node [above] {\scalebox{1}{$3$}} (s7);
        \draw[->] (s2) -> node [above] {\scalebox{1}{$2$}} (s8);
        \fill[cyan] (3.9cm, -1.2cm)--(3.9cm, 1.2cm)--(7.4cm, 1.2cm)--(7.4cm, -1.2cm)--cycle;
        \node[inner sep=0] at (7.2cm, 1.0cm) (B) {$\opg{B}$};
        \node[interface,fill=white, right = 0.5cm of exr1] (enr2) {\scalebox{0.8}{$\enrarg{2}$}};
        \node[innodeEve,right=0.75cm of enr2] (s3) {\scalebox{1}{$b$}};
        \node[interface,fill=white, right = 0.5cm of exr2] (enr3) {\scalebox{0.8}{$\enrarg{3}$}};
        \node[innodeEve,right=0.75cm of enr3] (s4) {\scalebox{1}{$c$}};
        \draw[->] (s7) -> (enr2);
        \draw[->] (enr2) -> node [above] {\scalebox{1}{$0$}} (s3);
        \draw[->] (s8) -> (enr3);
        \draw[->] (enr3) -> node [above] {\scalebox{1}{$0$}} (s4);
        \draw[->] (s3) edge[loop right] node[right] {$m_{1}$} (s3); 
        \draw[->] (s4) edge[loop right] node[right] {$m_{2}$} (s4); 
    \end{tikzpicture}
    }
    \caption{The parity game $\opg{A}\seqcomp \opg{B}$. Every node is owned by Player $\eve$. }
    \label{fig:multi-objective}
\end{figure}
For instance, \cref{fig:multi-objective} shows a parity game $\opg{A}\seqcomp \opg{B}$ consisting of open parity games $\opg{A}$ and $\opg{B}$.
In the open parity game $\opg{A}$, the Player $\eve$ has the choice of moving from the node $a$ to either the node $\exrarg{1}$ or the node $\exrarg{2}$. 
However, the optimal  $\eve$-strategies in $\opg{A}\seqcomp \opg{B}$ may vary, depending on the priorities $m_1$ and $m_2$ in the open parity game $\opg{B}$. 
Consequently, it is necessary to compute all  ``potentially'' optimal strategies on $\opg{A}$ without any prior knowledge of $\opg{B}$ in order to solve the parity game $\opg{A}\seqcomp \opg{B}$ compositionally. 
Our previous work on string diagrams of mean payoff games~\cite{DBLP:journals/corr/abs-2307-08034} tackled this by enumerating such potentially optimal strategies in open mean payoff games.
The enumeration algorithm requires exponential time in the number of nodes of the given open mean payoff game, which is infeasible in practice.  

By following our recent progress on string diagrams of Markov decision processes (MDPs)~\cite{WVHRJ2024accepted}, 
we introduce \emph{Pareto fronts} of open parity games, which are intuitively sets of optimal results on open parity games.
We prove \emph{positional determinacy} of open parity games, that is, positional strategies suffice to achieve Pareto fronts (\cref{thm:positionalDeterminacy}). 
Our proof method involves selecting a ``single-objective'', the number of which is finite, and 
we show that positional strategies suffice to satisfy the given single-objective.
This procedure is repeated for each possible single-objective. 
In the proof, we propose a novel translation that converts an open parity game into a (non-open) parity game tailored to a given single-objective (\cref{def:loop_construction}). 
Moreover, our translation leads to a simple algorithm for open parity games that requires exponential time only in the number of exits, not in the number of whole nodes as in our previous algorithm~\cite{DBLP:journals/corr/abs-2307-08034}: here our algorithm use known quasi-polynomial-time algorithms (\cref{alg:paretoCurveOPGs} and \cref{prop:complexityOPG}). 
On top of these developments, we provide a compositional algorithm of string diagrams of parity games by composing Pareto fronts of open parity games (\cref{alg:sdPGs}).

\myparagraph{Contributions:}
\begin{itemize}
    \item we introduce Pareto fronts of open parity games in \cref{subsec:pareto_curves_for_oPGs}, 
    \item we prove the positional determinacy of Pareto fronts of open parity games, which is stated in~\cref{subsec:positionalDeterminacy}, by a novel translation from an open parity game to a parity game with a given single-objective in \cref{subsec:loopConst},
    \item we present a simple algorithm that computes Pareto fronts of open parity games, and provide an upper bound of its time complexity in~\cref{subsec:alg_time_complexity}, which shows improvements over our previous algorithm~\cite{DBLP:journals/corr/abs-2307-08034},
    \item we provide a compositional algorithm of string diagrams of parity games in~\cref{sec:solvingSD}. 
\end{itemize}

 \subsection{Related Work}
\emph{String diagrams} are a graphical language, which is mostly described in category theory. 
They effectively model signal flow diagrams~\cite{DBLP:journals/jacm/BonchiGKSZ22,DBLP:conf/popl/BonchiSZ15,DBLP:journals/iandc/BonchiSZ17,baez2014categories}, quantum processes~\cite{DBLP:conf/icalp/CoeckeD08,DBLP:books/cu/CK2017}, Petri nets~\cite{DBLP:journals/pacmpl/BonchiHPSZ19,DBLP:conf/lics/BaezGMS21}, and games for economics~\cite{DBLP:conf/lics/GhaniHWZ18,DBLP:conf/csl/LavoreH021}.
Rathke et al.~\cite{DBLP:conf/rp/RathkeSS14} derive a divide-and-conquer algorithm for reachability analysis of Petri nets, and we provide a compositional algorithm for string diagrams of mean payoff games in~\cite{DBLP:journals/corr/abs-2307-08034}, which can be readily adapted to string diagrams of parity games~\cite{DBLP:journals/corr/abs-2112-14058}.
Recently, Piedeleu~\cite{Piedeleu25} provides a sound and complete diagrammatic language of parity games w.r.t. the semantics of open parity games~\cite{DBLP:journals/corr/abs-2112-14058}. 
In this paper, we propose a compositional algorithm for string diagrams of parity games based on a rather efficient semantics, focusing on the positional determinacy of open parity games.

\emph{Pareto fronts}~\cite{DBLP:conf/focs/PapadimitriouY00,DBLP:journals/lmcs/EtessamiKVY08} offer a well-formulated solution for multi-objective optimization problems.
They enhance the performance of compositional probabilistic model checking~\cite{WVHRJ2024accepted}.
A crucial property of Pareto fronts in~\cite{WVHRJ2024accepted} is again \emph{positional determinacy}: each point on Pareto fronts is induced by a positional scheduler, which can be computed using standard algorithms for (non-open) MDPs~\cite{DBLP:conf/atva/ForejtKP12}. 

In~\cite{WVHRJ2024accepted}, we propose a compositional probabilistic model checking of string diagrams of MDPs with Pareto fronts: we introduce \emph{shortcut MDPs} and prove their compositionality, analogous to~\cref{prop:base_Pareto_curve,prop:comp_Pareto_curve}.
A major contribution of this paper, compared to~\cite{WVHRJ2024accepted}, is the introduction of a novel translation from open parity games to parity games, designed to efficiently compute Pareto fronts.
This stands in contrast to the already established translation for open MDPs~\cite{DBLP:conf/atva/ForejtKP12}.

Multi-objectives in games have been well-studied in the literature.
Velner et al.~\cite{DBLP:journals/iandc/VelnerC0HRR15} study the complexity of multi-mean-payoff and multi-energy games. 
Chen et al.~\cite{DBLP:conf/mfcs/ChenFKSW13} develop an approximation method for Pareto fronts of stochastic games.
Recently, a Pareto-based rational verification problem, which interacts with the environment, is introduced in~\cite{DBLP:conf/concur/BruyereRT22}.

Obdr\v{z}\'{a}lek~\cite{DBLP:conf/cav/Obdrzalek03}, and Fearnley and Schewe~\cite{DBLP:journals/corr/abs-1112-0221} provide polynomial time algorithms for parity games on graphs of bounded tree-width.
Berwanger et al.~\cite{DBLP:conf/stacs/BerwangerDHK06} show that parity games on graphs of bounded DAG-width can be solved in polynomial time.
The notion of \emph{result} in~\cite{DBLP:conf/stacs/BerwangerDHK06} gathers optimal results of Player $\adam$ for each strategy employed by Player $\eve$, regardless of its optimality.  
In contrast, our Pareto fronts collect optimal results not only from Player $\adam$, but also from Player $\eve$, and are positionally determined.  
Although our compositional algorithm is not a polynomial time algorithm for string diagrams of parity games, it does not require assumptions such as bounded tree- or DAG-widths.

Compositionality in parity games seems to be linked to \emph{higher-order model checking}~\cite{DBLP:conf/lics/Ong06,DBLP:conf/lics/KobayashiO09}, a model checking method of higher-order programs. 
In higher-order model checking, specifications are given as parity tree automata, which are equivalent to modal $\mu$-calculus~\cite{DBLP:conf/focs/EmersonJ91}.
In fact, the semantics of string diagrams of parity games~\cite{DBLP:journals/corr/abs-2112-14058} is based on the semantics of higher-order model checking~\cite{DBLP:conf/mfcs/GrelloisM15}. 
Unlike~\cite{DBLP:conf/mfcs/GrelloisM15} that employs the trivial discrete order on priorities, we utilize the \emph{sub-priority order}~\cite{DBLP:conf/stacs/BerwangerDHK06,DBLP:conf/aplas/FujimaIK13}, which is central to our development. 
Fujima et al.~\cite{DBLP:conf/aplas/FujimaIK13} introduce the subtype relation in the intersection type, which corresponds to our order for results of games defined in~\cref{def:results_order}.  
Tsukada and Ong~\cite{DBLP:conf/csl/TsukadaO14} present a compositional higher-order model checking over the composition of B\"{o}hm trees. We employ the \emph{dual} of priorities (\cref{def:least_passed_priority}), which in~\cite{DBLP:conf/csl/TsukadaO14} is called the left-residual of $0$.
The duality of queries shown in~\cref{lem:dualQuery} is conceptually similar to the notion of the complementarity predicate in~\cite{DBLP:journals/iandc/SalvatiW14}. 
They utilize the property for showing the relationship between a game and its dual in the context of higher-order model checking.

\section{Formal Problem Statement}
 We revisit parity games and string diagrams of parity games~\cite{DBLP:journals/corr/abs-2112-14058}, and we formally state our target problem. 
 Throughout this paper, we fix the finite set $\setPriorities \defeq \{0, 1, \dots, M\}$ of \emph{priorities}. We assume that the maximum priority $M$ is even and $M\geq 2$ without loss of generality. 

\subsection{Parity Games}
\begin{definition}[parity game]
A \emph{parity game (PG)} $\pg{G}$ is a tuple $(V, V_{\eve}, V_{\adam}, E, \Omega)$, 
where the set $V$ is a non-empty finite set of \emph{nodes},
the sets $V_{\eve}, V_{\adam}$ are a partition of $V$ such that $V_{\eve}$ is owned by Player $\eve$ and $V_{\adam}$ is owned by Player $\adam$,
the set $E$ is a set of \emph{edges} $E\subseteq V\times V$, 
and the function $\Omega$ is a \emph{labelling function} $\Omega\colon E\rightarrow \setPriorities$, assigning priorities to edges.  
 \end{definition}
We assume that every node has a successor in PGs. 
A \emph{play} $p$ is an infinite path $(v_i)_{i\in \nat}$ such that $(v_i, v_{i+1})\in E$, for any $i\in \nat$.
A play $p$ is \emph{winning (for Player $\eve$)} if the maximum priority that occurs in $p$ infinitely often is even.
Otherwise, the play $p$ is \emph{losing (for Player $\eve$)}. 
We also call the winning condition as the \emph{parity condition}. 
An \emph{$\eve$-strategy $\sigma_{\eve}$ (of Player $\eve$)} is a function $\sigma_{\eve}\colon V^{\ast}\cdot V_{\eve} \rightarrow V$ such that for any finite path $v_1\cdots v_n$, if $v_n\in V_{\eve}$, then $\big(v_n,\sigma_{\eve}(v_1\cdots v_n)\big) \in E$.
A \emph{$\adam$-strategy $\sigma_{\adam}$ (of Player $\adam$)} is analogously defined. 
An $\eve$-strategy $\sigma_{\eve}$ is a \emph{positional strategy} if for any $v_1\cdots v_m$ and $v'_1\dots v'_n$ such that $v_m = v'_n\in V_{\eve}$, their successors coincide $\sigma_{\eve}(v_1\cdots v_m) = \sigma_{\eve}(v'_1\cdots v'_n)$; \emph{positional $\adam$-strategies} are defined in the same way. 
We write $\sigma_{\eve}\colon V_{\eve}\rightarrow V $ (and $\sigma_{\adam}\colon V_{\adam}\rightarrow V $) for a positional $\eve$-strategy (and a positional $\adam$-strategy).
Given a pair $(\sigma_{\eve}, \sigma_{\adam})$ of $\eve$- and $\adam$-strategies, and a node $v$, a \emph{play $\indPlay{\sigma_{\eve}}{\sigma_{\adam}}{v}$ from $v$ that is induced by $(\sigma_{\eve}, \sigma_{\adam})$} is the play $(v_i)_{i\in \nat} $ such that $v_1 = v$,
and  for any index $i$, if $v_i\in V_{\eve}$, then $\sigma_{\eve}(v_1\cdots v_i) = v_{i+1}$, and  if $v_i\in V_{\adam}$, then $\sigma_{\adam}(v_1\cdots v_i) = v_{i+1}$. 
PGs are \emph{positionally determined}~\cite{DBLP:conf/lop/Emerson85,DBLP:conf/focs/EmersonJ91}, that is, there is a partition $(W_{\eve}, W_{\adam})$ of $V$, called \emph{winning regions}, such that $v\in W_{\eve} $ iff  there is a winning \emph{positional} strategy $\sigma_{\eve}\colon V_{\eve}\rightarrow V$ for Player $\eve$, and $v\in W_{\adam}$ iff there is a winning positional strategy $\sigma_{\adam}\colon V_{\adam}\rightarrow V$ for Player $\adam$.

\subsection{String Diagrams of Parity Games}
\label{subsec:sd_pgs}
String diagrams of parity games~\cite{DBLP:journals/corr/abs-2112-14058} consist of \emph{open parity games (oPGs)}, 
which are a compositional extension of PGs. OPGs have \emph{open ends} that are \emph{entrances} and \emph{exits}.
We illustrate an oPG in~\cref{fig:running_example} as a running example, and 
will continue to refer to this example until~\cref{subsec:positionalDeterminacy}. 
It is important to note that oPGs without exits can be regarded as PGs.
\begin{definition}[open parity game~\cite{DBLP:journals/corr/abs-2112-14058}]
\label{def:open_parity_game}
An \emph{open parity game (oPG)} $\opg{A}$ is a pair $(\pg{G}, \interface)$ of a PG $\pg{G} = (V, V_{\eve}, V_{\adam}, E, \Omega)$ and open ends $\interface = (\enr, \enl, \exr, \exl) $,
where $\enr, \enl, \exr, \exl$ are nodes $\enr, \enl, \exr, \exl\subseteq V_{\eve}$ that are totally ordered, respectively, and each pair of them is disjoint.
Open ends in $\enr, \exr$ are \emph{rightward}, and open ends in $\enl, \exl$ are \emph{leftward}. 
We call $\en \defeq \enr\cup \enl$ \emph{entrances}, and $\ex \defeq \exr\cup \exl$ \emph{exits}, respectively.
We also write $\opg{A}\colon (m_{\dr}, m_{\dl})\rightarrow (n_{\dr}, n_{\dl})$ for the oPG $\opg{A}$ with the type of open ends, where $m_{\dr}\defeq |\enr|$, $m_{\dl}\defeq |\enl|$, $n_{\dr}\defeq |\exr|$, and $n_{\dl}\defeq |\exl|$.
We additionally impose that every exit $o\in\exr\cup \exl$ is a sink node, that is, for any $v\in V$, $(o, v)\in E$ implies $v = o$ and $\Omega(o, o) = 0 $.  
\end{definition}

\begin{figure}[t]
    \centering
    \scalebox{1.4}{
    \begin{tikzpicture}[
innodeEve/.style={draw, circle, minimum size=0.5cm,fill=white},
innodeAdam/.style={draw, diamond, minimum size=0.5cm,fill=white},
interface/.style={draw, rectangle, minimum size=0.5cm,fill=white},]
        \fill[cyan] (-0.6cm, -1cm)--(-0.6cm, 1.2cm)--(4.6cm, 1.2cm)--(4.6cm, -1cm)--cycle;
        \node[inner sep=0] at (4.4cm, 1.0cm) (A) {$\opg{C}$};
        \node[interface,fill=white, yshift=0.5cm] (s1) {\scalebox{0.8}{$\enrarg{1}$}};
        \node[interface,fill=white, yshift=-0.5cm] (s9) {\scalebox{0.8}{$\exlarg{1}$}};
        \node[inner sep=0,right=-1.3cm of s9] (exl1) {};
        \node[inner sep=0,right=-1.3cm of s1] (enr1) {};
        \node[innodeEve,right=0.75cm of s1] (s2) {\scalebox{0.8}{$a$}};
        \node[innodeAdam,right=0.75cm of s9] (s3) {\scalebox{0.8}{$b$}};
        \node[innodeEve,right=0.75cm of s2] (s4) {\scalebox{0.8}{$c$}};
        \node[innodeAdam,right=0.75cm of s3] (s5) {\scalebox{0.8}{$d$}};
        \node[interface,fill=white,right=0.75cm of s4, yshift=-0.4cm] (s7) {\scalebox{0.8}{$\exrarg{1}$}};
        \node[inner sep=0,right=0.6cm of s7] (exr1) {};
        \draw[->] (enr1) -> (s1);
        \draw[->] (s9) -> (exl1);
        \draw[->] (s1) -> node [above] {\scalebox{0.8}{$0$}} (s2);
        \draw[->] (s3) -> node [above] {\scalebox{0.8}{$0$}} (s9);
        \draw[->] (s2) -> node [above] {\scalebox{0.8}{$1$}} (s4);
        \draw[->] (s4) -> node [left] {\scalebox{0.8}{$0$}} (s5);
        \draw[->] (s4) -> node [above] {\scalebox{0.8}{$0$}} (s7);
        \draw[->] (s5) -> node [above] {\scalebox{0.8}{$1$}} (s3);
        \draw[->] (s3) -> node [left] {\scalebox{0.8}{$3$}} (s2);
        \draw[->] (s5) -> node [above] {\scalebox{0.8}{$2$}} (s7);
        \draw[->] (s7) -> (exr1);
    \end{tikzpicture}
    }
    \caption{An oPG $\opg{C}\colon (1, 1)\rightarrow (1, 0)$. The node $i_{\dr, 1}$ is the entrance, and the nodes $o_{\dr, 1}, o_{\dl, 1}$ are exits. The open ends $\enrarg{1}$ and $\exrarg{1}$ are rightward, and the exit $\exlarg{1}$ is leftward.  The shape of each (internal) node represents the owner, that is, circles are owned by Player $\eve$, and diamonds are owned by Player $\adam$. The labels on the edges are the assigned priorities. }
    \label{fig:running_example}
\end{figure} 

Algebraic operations, namely the \emph{(bidirectional) sequential composition} $\opg{A}\seqcomp \opg{B}$ and the \emph{sum} $\opg{A}\oplus \opg{B}$,  are defined as illustrated in~\cref{fig:string_diagrams}. 
The sequential composition $\opg{A}\seqcomp \opg{B}$ of $\opg{A}$ and $\opg{B}$ is given by connecting their open ends, assuming their arities are consistent. 
Note that the sequential composition can create a cycle along open ends. 
\begin{definition}[sequential composition]
\label{def:seqOPG}
Let $\opg{A}\colon (m_{\dr}, m_{\dl})\rightarrow (l_{\dr}, l_{\dl})$ and $\opg{B}\colon (l_{\dr}, l_{\dl})\rightarrow (n_{\dr}, n_{\dl})$. Their \emph{(bidirectional) sequential composition} $\opg{A}\seqcomp \opg{B}\colon (m_{\dr}, m_{\dl})\rightarrow (n_{\dr}, n_{\dl})$ is an oPG $(V^{\opg{A}}\uplus V^{\opg{B}}, V^{\opg{A}}_{\eve}\uplus V^{\opg{B}}_{\eve}, V^{\opg{A}}_{\adam}\uplus V^{\opg{B}}_{\adam}, E^{\opg{A}\seqcomp \opg{B}}, \Omega^{\opg{A}\seqcomp \opg{B}})$, where the set $E^{\opg{A}\seqcomp \opg{B}}$ of edges is the least relation and $\Omega^{\opg{A}\seqcomp \opg{B}}$ is the function that satisfy the following conditions: 
\begin{itemize}
    \item if $(v, v')\in E^{\opg{A}}$ and $v\not \in \exr^{\opg{A}}$, then $(v, v')\in E$ and $\Omega^{\opg{A}\seqcomp \opg{B}}(v, v') \defeq \Omega^{\opg{A}}(v, v')$,
    \item if $(v, v')\in E^{\opg{B}}$ and $v\not \in \exl^{\opg{B}}$, then $(v, v')\in E$ and $\Omega^{\opg{A}\seqcomp \opg{B}}(v, v') \defeq \Omega^{\opg{B}}(v, v')$, 
    \item $(\exrarg{i}^{\opg{A}}, \enrarg{i}^{\opg{B}})\in E$ and $\Omega^{\opg{A}\seqcomp \opg{B}}(\exrarg{i}^{\opg{A}}, \enrarg{i}^{\opg{B}}) \defeq 0$, for any $i\in [1, l_{\dr}]$, and
    \item $(\exlarg{j}^{\opg{B}}, \enlarg{j}^{\opg{A}})\in E$ and $\Omega^{\opg{A}\seqcomp \opg{B}}(\exlarg{j}^{\opg{B}}, \enlarg{j}^{\opg{A}}) \defeq 0$, for any $j\in [1, l_{\dl}]$.
\end{itemize}    
\end{definition}

We also define the \emph{sum} $\opg{A}\oplus \opg{B}$ of $\opg{A}$ and $\opg{B}$: see~\cref{fig:string_diagrams} for an illustration. 

\begin{definition}[sum]
\label{def:sumOPG}
Let $\opg{A}\colon (m_{\dr}, m_{\dl})\rightarrow (n_{\dr}, n_{\dl})$ and $\opg{B}\colon (k_{\dr}, k_{\dl})\rightarrow (l_{\dr}, l_{\dl})$. Their \emph{sum} $\opg{A}\oplus \opg{B}\colon (m_{\dr}+ k_{\dr}, m_{\dl}+k_{\dl})\rightarrow (n_{\dr} + l_{\dr}, n_{\dl} + l_{\dl})$ is an oPG $(V^{\opg{A}}\uplus V^{\opg{B}}, V^{\opg{A}}_{\eve}\uplus V^{\opg{B}}_{\eve}, V^{\opg{A}}_{\adam}\uplus V^{\opg{B}}_{\adam}, E^{\opg{A}\oplus \opg{B}}, \Omega^{\opg{A}}\uplus \Omega^{\opg{B}})$, where the set $E^{\opg{A}\oplus \opg{B}}$ of edges is the least relation that satisfies the following condition: 
\begin{itemize}
    \item if $(v, v')\in E^{\opg{A}}$, then $(v, v')\in E^{\opg{A}\oplus \opg{B}}$, and if $(v, v')\in E^{\opg{B}}$, then $(v, v')\in E^{\opg{A}\oplus \opg{B}}$.
\end{itemize}
\end{definition}

We define string diagrams of PGs as a syntax of compositional PGs: 
we also introduce \emph{operational semantics} of string diagrams of PGs
as oPGs induced by string diagrams.  

\begin{definition}[string diagram]
\label{def:sd_opgs}
A \emph{string diagram $\sd{D}$ of parity games} is a term inductively generated by
 $\sd{D} ::= \opg{A} \mid \sd{D} \seqcomp \sd{D} \mid \sd{D} \oplus \sd{D}$
where $\opg{A}$ ranges over oPGs.
The \emph{operational semantics} $\semantics{\sd{D}}$ of a string diagram $\sd{D}$ is the oPG inductively given by~\cref{def:seqOPG,def:sumOPG}.
\end{definition}

Finally, we state the target problem, that is, deciding the winner for each entrance on a given string diagram $\sd{D}$ 
whose operational semantics is a PG (oPG that has no exits). 
\begin{mdframed}
\textbf{Problem Statement:}
Let $\sd{D}$ be a string diagram such that the operational semantics $\semantics{\sd{D}}$ can be naturally considered as a PG, that is, $\semantics{\sd{D}}$ has no exits. Compute the winning region $(W_{\eve}, W_{\adam})$ on entrances $\en^{\semantics{\sd{D}}}$, that is, decide the winner of each entrance.   
\end{mdframed}

\begin{remark}[trace operator]
    In~\cite{DBLP:journals/corr/abs-2112-14058,DBLP:conf/cav/WatanabeEAH23,DBLP:journals/corr/abs-2307-08034}, we support the trace operator that makes loops by connecting exits to the corresponding entrances, respectively. 
    We can encode the trace operator by the (bidirectional) sequential composition and the sum, as in~\cite{WVHRJ2024accepted,WatanabeVJH24}. 
    
\end{remark}

\section{Denotations, Optimal Strategies, and Pareto Fronts}
\label{sec:denotation_paretoCurve}

In this section, we recap \emph{denotations}~\cite{DBLP:journals/corr/abs-2112-14058} and introduce \emph{Pareto fronts} for oPGs, which are a compositional extension of winning and losing for PGs.

\subsection{Denotation of Open Parity Games}
\emph{Strategies} and \emph{plays} in oPGs are defined in exactly the same way as in PGs; note that oPGs are PGs with open ends.
Unlike PGs, where the winning condition is \emph{qualitative}, the winning condition in oPGs is \emph{quantitative} due to  the openness of oPGs, which induces \emph{pending states} between two extremes of winning and losing.
Specifically, we show denotations of oPGs~\cite{DBLP:journals/corr/abs-2112-14058} with \emph{denotations} of plays and entrances.
Given a play $p$, the denotation $\semPlay{p}$ of the play $p$ is \emph{winning} ($\top$), \emph{losing} ($\bot$), or reaching an exit $o$ with the priority $m$ that is the maximum priority that occurred during the play.   
Importantly, all exits are sinks, meaning that once we enter an exit, we stay there forever.
\begin{definition}[domain]
    \label{def:domain}
    Let $\opg{A}$ be an oPG. The \emph{domain} $\domain{\opg{A}}$ of $\opg{A}$ is the set  $\domain{\opg{A}}\defeq \{ \bot, \top \}\cup \{ (o, m) \mid \text{$o$ is an exit $(o\in \ex^{\opg{A}})$},\ m \text{ is a priority } (m\in \setPriorities) \}$. 
\end{definition}

\begin{definition}[denotation of play]
    Let $p = (v_j)_{j\in \nat} $ be a (necessarily infinite) play on an oPG $\opg{A}$. The \emph{denotation} $\semPlay{p}\in\domain{\opg{A}}$ of $p$ is given by the following: 
    \begin{align*}
        \semPlay{p}\defeq \begin{cases}
            \top &\text{ if $v_j\not\in \ex^{\opg{A}}$ for any $j\in\nat$, and $p$ satisfies the parity condition,}\\
            (o, m) &\text{ if $v_j = o\in \ex^{\opg{A}}$ for some $j\in \nat$,}\\
            &\text{ and $m \defeq \max\big\{ \Omega\big((v_k, v_{k+1})\big) \ \big |\   k\in \{1,\dots, j-1\} \big\}$, }\\
            \bot &\text{ if $v_j\not\in \ex^{\opg{A}}$ for any $j\in\nat$, and $p$ does not satisfy the parity condition.}
        \end{cases}
    \end{align*}
    Note that the condition in the second case distinction is well-defined
    because we assume every exit is a sink state (\cref{def:open_parity_game}).  
\end{definition}

The denotation $\semEntrance{i}$ of an entrance $i$ is the set of denotations of plays from $i$ that are induced by $\eve$- and $\adam$-strategies.
Specifically, for each $\eve$-strategy $\sigma_{\eve}$, the denotation $\semEntrance{i}$ contains the set of denotations $\semPlay{\indPlay{\sigma_{\eve}}{\sigma_{\adam}}{i}}$ of plays $\indPlay{\sigma_{\eve}}{\sigma_{\adam}}{i}$ that are induced by some $\adam$-strategy $\sigma_{\adam}$. 
Note that $\sigma_{\eve}$ and $\sigma_{\adam}$ may not be positional. 
\begin{definition}[denotation of entrance]
    \label{def:denEntrance}
Let $\opg{A}$ be an oPG and $i$ be an entrance. The \emph{denotation} $\semEntrance{i}$ is given by 
\begin{align*}
    \semEntrance{i}\defeq \Big\{ \big\{\semPlay{\indPlay{\sigma_{\eve}}{\sigma_{\adam}}{i}} \ \big | \ \text{$\sigma_{\adam}$ is a $\adam$-strategy}\big\} \ \Big | \ \text{$\sigma_{\eve}$ is an $\eve$-strategy}\Big\}.
\end{align*}
We call the entrance $i$ \emph{winning} if $\{\top\}\in \semEntrance{i}$, \emph{losing} if $\bot \in T$ holds for any $T\in \semEntrance{i}$, and \emph{pending} otherwise. 
\end{definition}

The denotations of entrances are indeed a compositional extension of winning regions of PGs: in fact, for oPG $\opg{A}$ with no exits, the entrance is winning if $\{\top\}\in \semEntrance{i}$, which is consistent with winning regions on PGs, and otherwise the entrance is losing. 

\begin{definition}[denotation of oPGs~\cite{DBLP:journals/corr/abs-2112-14058}]
    \label{def:denoOPGs}
    Let $\opg{A}$ be an oPG. The \emph{denotation} $\semOPGs{\opg{A}}$ of $\opg{A}$ is the indexed family $\big(\semEntrance{i}\big)_{i\in \en^{\opg{A}}}$ of denotations of entrances. 
\end{definition}

\begin{example}
Consider the oPG $\opg{C}\colon (1, 1)\rightarrow (1, 0)$ presented in~\cref{fig:running_example}.
Since there is the unique entrance $\enrarg{1}$, the denotation $\semOPGs{\opg{C}}$ is essentially same as the denotation $\semEntrance{\enrarg{1}}$.
Every play that starts from the entrance reaches the node $c$:
 if Player $\eve$ decides to move the exit $\exrarg{1}$, then the denotation of the play is $(\exrarg{1}, 1)$, and otherwise, the play reaches the node $d$.
By similar arguments, we can conclude that the denotation  $\semEntrance{\enrarg{1}}$ is the following: 
\begin{align*}
    \semEntrance{\enrarg{1}}=\Big\{ &\big\{ (\exrarg{1}, 1)\big\},\, \big\{ (\exrarg{1}, 2), (\exrarg{1}, 3), (\exlarg{1}, 1),(\exlarg{1}, 3), \bot \big\},\,
    \big\{ (\exrarg{1}, 2), (\exrarg{1}, 3), (\exlarg{1}, 1)\big\},\,\\
    &\big\{ (\exrarg{1}, 2), (\exrarg{1}, 3), (\exlarg{1}, 1), (\exlarg{1}, 3)\big\}
    \Big\}.
\end{align*}
\end{example}

\begin{remark}
   Our denotation of entrances implicitly requires that an $\eve$-strategy $\sigma_{\eve}$ is always chosen at first, and later a $\forall$-strategy is chosen with respect to $\sigma_{\eve}$. 
   This order of choices by Players $\eve$ and $\adam$ does not matter for parity games because parity games are positionally determined~\cite{DBLP:conf/lop/Emerson85,DBLP:conf/focs/EmersonJ91}.
\end{remark}

\subsection{Optimal Strategies of Open Parity Games}
\label{subsec:pareto_curves_opgs}
For deciding the winning region of the entrances in the operational semantics $\semantics{\sd{D}}$,
as stated in~\cref{subsec:sd_pgs}, the denotation $\semOPGs{\opg{A}}$ defined in~\cref{def:denoOPGs} may contain redundant information:
for example, suppose we fix an $\eve$-strategy $\sigma_{\eve}$.
If there is a $\adam$-strategy $\sigma_{\adam}$ such that $\semPlay{\indPlay{\sigma_{\eve}}{\sigma_{\adam}}{i}}\defeq \bot$, that is,
the play $\indPlay{\sigma_{\eve}}{\sigma_{\adam}}{i}$ is losing for Player $\eve$, then we can exclude other $\adam$-strategies $\tau_{\adam}$ because choosing $\sigma_{\adam}$ is optimal for Player $\adam$ when Player $\eve$ chooses $\sigma_{\eve}$. 
We generalize this elimination of sub-optimal strategies by introducing orders on plays and strategies.

Specifically, we use the \emph{sub-priority order}~\cite{DBLP:conf/aplas/FujimaIK13,DBLP:conf/stacs/BerwangerDHK06} on priorities: the intuition is that Player $\eve$ favors larger even priorities, and Player $\adam$ favors larger odd priorities. 

\begin{definition}[sub-priority order~\cite{DBLP:conf/aplas/FujimaIK13,DBLP:conf/stacs/BerwangerDHK06}]
\label{def:sub_priority_order}
The \emph{sub-priority order} $\priorityOrder$ on the set $\setPriorities$ of priorities is the total order defined by $M-1\priorityStOrder M-3 \priorityStOrder \dots \priorityStOrder 1 \priorityStOrder 0 \priorityStOrder 2  \priorityStOrder \dots \priorityStOrder M-2  \priorityStOrder M$.
Recall that we assume that the maximum priority $M$ is even. 
\end{definition}
\begin{remark}
    The existing semantics~\cite{DBLP:conf/mfcs/GrelloisM15,DBLP:journals/corr/abs-2112-14058} uses the trivial discrete order on the set   $\setPriorities$, that is, $m_1\priorityOrder m_2$ iff $m_1 = m_2$. 
    This difference is crucial for our development. 
\end{remark}

The operation $\max\colon \setPriorities\times \setPriorities\rightarrow \setPriorities$ that returns a bigger number with respect to the standard order on natural numbers is (also) monotone 
with respect to the sub-priority order. 

\begin{lemma}[monotonicity~\cite{DBLP:conf/csl/TsukadaO14}]
    \label{lem:monoPriority}
    Let $m_1, m_2, m_3\in \setPriorities$. 
    If the inequality $m_1\priorityOrder m_2$ holds, then the inequality $\max(m_1, m_3)\priorityOrder \max(m_2, m_3)$ holds. 
\end{lemma}

We extend the sub-priority order to a partial order on the domain $\domain{\opg{A}}$ considering exits. 
\begin{definition}[domain order]
    The \emph{domain order} $\domainOrder$ on the domain $\domain{\opg{A}}$ (defined in~\cref{def:domain}) is the partial order such that $ d_1\domainOrder d_2$ iff (i) $d_1 = \bot$, (ii) $d_2 = \top$, or (iii) $d_1\defeq (o_1, m_1)$ and $d_2\defeq (o_2, m_2)$ such that $o_1 = o_2$ and $m_1\priorityOrder m_2$.
\end{definition}

Given an $\eve$-strategy $\sigma_{\eve}$,  we define \emph{optimal $\adam$-strategies} w.r.t. $\sigma_{\eve}$ as $\adam$-strategies that induce plays whose denotations are minimal.   
\begin{definition}[minimal]
    For a set $T\subseteq \domain{\opg{A}}$, the set $\worst{T}$ is the set of minimal elements in $T$ w.r.t. the domain order $\domainOrder$. 
\end{definition}

\begin{definition}[optimal $\adam$-strategy]
    Let $\opg{A}$ be an oPG, $i$ be an entrance, and $\sigma_{\eve}$ be an $\eve$-strategy.
    A $\adam$-strategy $\sigma_{\adam}$ is \emph{optimal} on $i$ w.r.t. $\sigma_{\eve}$ if the following condition holds: 
    \begin{equation*}
        \semPlay{\indPlay{\sigma_{\eve}}{\sigma_{\adam}}{i}} \in \worst{\big\{ \semPlay{\indPlay{\sigma_{\eve}}{\tau_{\adam}}{i}}  \  \big | \ \tau_{\adam}\text{ is a $\adam$-strategy}\big\}}. 
    \end{equation*}
\end{definition}

We further define optimal $\eve$-strategies; here we do not fix a $\adam$-strategy. 
Specifically, we use the \emph{upper preorder}~\cite{DBLP:conf/lics/Abramsky87}.
\begin{definition}[upper preorder~\cite{DBLP:conf/lics/Abramsky87}]
    The \emph{upper preorder} $\upperOrder$ on the powerset $\mathcal{P}(\domain{\opg{A}})$ of $\domain{\opg{A}}$ is given by 
    $T_1\upperOrder T_2$ if for any $d_2\in T_2$, there is $d_1\in T_1$ such that $d_1\domainOrder d_2$. 
\end{definition}
The intuition is the following: we fix an entrance $i$ and suppose that there are two $\eve$-strategies $\sigma_{\eve}$ and $\tau_{\eve}$. 
The $\eve$-strategy $\sigma_{\eve}$ is ``better'' than $\tau_{\eve}$ 
if for any $\adam$-strategy $\sigma_{\adam}$, there is a $\adam$-strategy $\tau_{\adam}$ such that 
$\indPlay{\tau_{\eve}}{\tau_{\adam}}{i} \playOrder \indPlay{\sigma_{\eve}}{\sigma_{\adam}}{i}$. 

\begin{definition}[maximal]
   For a set $S\subseteq \mathcal{P}(\domain{\opg{A}})$, the set 
$\best{S}$ is the set of maximal elements in $S$ w.r.t. the upper preorder $\upperOrder$.     
\end{definition}

\begin{definition}[optimal $\eve$-strategy]
    Let $\opg{A}$ be an oPG, and $i$ be an entrance.
    An $\eve$-strategy $\sigma_{\eve}$ is \emph{optimal} on $i$ if the following condition holds: 
    \begin{equation*}
        \worst{\big\{ \semPlay{\indPlay{\sigma_{\eve}}{\sigma_{\adam}}{i}}  \  \big | \ \text{a $\adam$-strategy }\sigma_{\adam}\big\}}\in  \best{\Big\{\worst{\big\{ \semPlay{\indPlay{\tau_{\eve}}{\tau_{\adam}}{i}}  \  \big | \ \text{a $\adam$-strategy }\tau_{\adam}\big\}} \ \Big | \ \text{an $\eve$-strategy } \tau_{\eve} \Big\}}.
    \end{equation*}
\end{definition}

We show an example of optimal $\eve$-strategies in~\cref{ex:paretoCurve} later.

\subsection{Pareto Fronts of Open Parity Games}
\label{subsec:pareto_curves_for_oPGs}
We conclude this section by introducing \emph{Pareto fronts} of oPGs.
Roughly speaking, Pareto fronts are the set of \emph{results} that are induced by optimal strategies.  

\begin{definition}[results, result order~\cite{DBLP:conf/aplas/FujimaIK13}]
    \label{def:results_order}
    Let $\opg{A}$ be an oPG. The set $\setPoints{\opg{A}}$ of \emph{results} is the set of non-empty antichains $\result$ on the domain $\domain{\opg{A}}$, that is, $\result \subseteq \domain{\opg{A}}$ and $\worst{\result} = \result$. 
    The \emph{result order} $\resultOrder{\opg{A}}$ on $\setPoints{\opg{A}}$ is the partial order that is given by the upper preorder $\upperOrder$.
\end{definition}
The result order is indeed a partial order because of the minimality of results. 
\begin{proposition}
    \label{prop:antisymmetry}
    The result order is a partial order. 
    \end{proposition}
\begin{proof}
Suppose that $\result\resultOrder{\opg{A}}\result'$ and $\result'\resultOrder{\opg{A}}\result$. 
It suffices to prove that $\result \subseteq \result'$ due to the symmetry of $\result$ and $\result'$. 
For any $d_1\in \result$, there is $d_2\in \result'$ such that $d_2 \domainOrder d_1$. 
Similarly, there is $d_3\in \result$ such that $d_3\domainOrder d_2$. 
Thus, we have $d_3 \domainOrder d_1$ and we see $d_3 = d_1$ due to the 
minimality of $\result$. Thus, $d_1 = d_2 = d_3$ holds and $\result \subseteq \result'$ holds. 
\end{proof}
We finally define \emph{Pareto fronts} of oPGs by optimal $\eve$- and $\adam$-strategies. 
\begin{definition}[Pareto front, Pareto-optimal]
Let $\opg{A}$ be an oPG, and $i$ be an entrance. The \emph{Pareto front} $\paretoCurve{\opg{A}}{i}$ is the set of denotations of plays that are induced by optimal $\eve$-strategies and optimal $\adam$-strategies on $i$, that is, 
\begin{align*}
    \paretoCurve{\opg{A}}{i}\defeq \best{\Big\{ \worst{\big\{ \semPlay{\indPlay{\sigma_{\eve}}{\sigma_{\adam}}{i}}  \  \big | \ \sigma_{\adam}\text{ is a $\adam$-strategy}\big\}} \ \Big |\  \sigma_{\eve}\text{ is an $\eve$-strategy} \Big\}} \subseteq \setPoints{\opg{A}}. 
\end{align*}
We call results $\result $ in $\paretoCurve{\opg{A}}{i}$ \emph{Pareto-optimal}. 
We write $\paretoCurves{\opg{A}}$ for the indexed family $\big(\paretoCurve{\opg{A}}{i}\big)_{i\in \en^{\opg{A}}}$ of Pareto fronts.  Note that $\paretoCurve{\opg{A}}{i} \not = \emptyset$ and $\emptyset\not \in \paretoCurve{\opg{A}}{i}$ hold by definition. 
\end{definition}

\begin{example}
    \label{ex:paretoCurve}
Consider the oPG $\opg{C}\colon (1, 1)\rightarrow (1, 0)$ presented in~\cref{fig:running_example}. 
The denotation $\semEntrance{\enrarg{1}}$ contains $\big\{ (\exrarg{1}, 2), (\exrarg{1}, 3), (\exlarg{1}, 1)\big\}$. 
Since $(\exrarg{1}, 3)\domainOrder (\exrarg{1}, 2)$, we have a result $\worst{\big\{ (\exrarg{1}, 2), (\exrarg{1}, 3), (\exlarg{1}, 1)\big\}} = \big\{ (\exrarg{1}, 3), (\exlarg{1}, 1)\big\}$.
This result is not Pareto-optimal because the inequality $ \big\{ (\exrarg{1}, 3), (\exlarg{1}, 1)\big\} \upperOrder \big\{ (\exrarg{1}, 1)\big\}$ holds, 
where  $\big\{ (\exrarg{1}, 1)\big\}\in \semEntrance{\enrarg{1}}$. 
By continuing similar arguments, we can show that the Pareto front $\paretoCurve{\opg{C}}{\enrarg{1}}$ is $
  \Big\{ \big\{ (\exrarg{1}, 1)\big\}
    \Big\}.
$
Thus, an optimal $\eve$-strategy on $\opg{C}$ is positionally moving from $c$ to $\exrarg{1}$. 
\end{example}

\begin{remark}[meager semantics]
In~\cite{DBLP:journals/corr/abs-2307-08034}, we introduce \emph{meager semantics} to eliminate sub-optimal strategies.
Meager semantics is defined \emph{globally}, meaning that optimal strategies are uniformly defined over entrances.
Pareto fronts are defined \emph{locally}, with optimal strategies tailored specifically for each entrance. This locality is essential for our technical development in~\cref{sec:solving_oPGs,sec:solvingSD}. 
    
\end{remark}

Pareto fronts of oPGs characterise the winning condition of oPGs defined in~\cref{def:denEntrance}.
\begin{lemma}
\label{lem:prop_paretoCurve}
Let $\opg{A}$ be an oPG, and $i$ be an entrance. 
The Pareto front $\paretoCurve{\opg{A}}{i}$ satisfies the following properties: 
\begin{enumerate}
    \item the entrance $i$ is winning iff  $\paretoCurve{\opg{A}}{i}=  \big\{\{\top\}\big\}$, 
    \item the entrance $i$ is losing iff $\paretoCurve{\opg{A}}{i} =  \big\{\{\bot\}\big\}$, and 
    \item the entrance $i$ is pending iff for any $T\in \paretoCurve{\opg{A}}{i}$, $\top, \bot\not\in T$, and  for any $(o_1, m_1), (o_2, m_2)\in T$, if  $o_1 = o_2$, then $m_1 = m_2$. 
\end{enumerate}
\end{lemma}
\begin{proof}

    Suppose that the entrance $i$ is winning. By definition, $\{\top\}\in \semEntrance{i}$ holds. 
    Since $\{\top\}$ is the greatest element, $\paretoCurve{\opg{A}}{i}=\big\{\{\top\}\big\}$. 
    
    Suppose that the entrance $i$ is losing. 
    By definition, $\bot \in T$ holds for any $T\in \semEntrance{i}$.
    This means that for any optimal $\adam$-strategy $\sigma_{\adam}$ of a given $\eve$-strategy $\sigma_{\eve}$, the denotation $\semPlay{\indPlay{\sigma_{\eve}}{\sigma_{\adam}}{i}}$ is $\bot$ because $\bot$ is the least element in the domain, which proves $\paretoCurve{\opg{A}}{i}=\big\{\{\bot\}\big\}$.
    
    Suppose that the entrance $i$ is pending. 
    For any $T\in \paretoCurve{\opg{A}}{i}$, $\top\not\in T$  because $\top$ is the greatest element. 
    Similarly, $\bot\not\in T$ because $T = \{\bot\}$ leads to the contradiction. If $(o, m_1), (o, m_2)\in T$, $m_1 = m_2$ holds because of the minimality of $T$. 
    \end{proof}

By the characterisation shown in~\cref{lem:prop_paretoCurve}, 
it suffices to compute the Pareto fronts of operational semantics of string diagrams of parity games for solving our target problem shown in~\cref{subsec:sd_pgs}. 
We thus extend our target problem as follows: 
\begin{mdframed}
    \textbf{Compositional Problem Statement:}
    Let $\sd{D}$ be a string diagram of PGs, and $i$ be an entrance. Compute the Pareto front $\paretoCurve{\semantics{\sd{D}}}{i}$ of the operational semantics $\semantics{\sd{D}}$.   
\end{mdframed}

\section{Solving Open Parity Games}
\label{sec:solving_oPGs}
In this section, we prove the positional determinacy of Pareto fronts through a novel translation from an oPG to a PG specified by a given \emph{query}, a method we call the \emph{loop construction}.  
Utilizing the loop construction, we provide a simple algorithm for computing Pareto fronts. 
Additionally, we present an analysis of the upper-bound of its time complexity. 

\subsection{Positional Determinacy of Pareto Fronts}
\label{subsec:positionalDeterminacy}

We begin by stating the positional determinacy of Pareto fronts: positional strategies suffice to determine Pareto fronts. 

\begin{theorem}[positional determinacy]
\label{thm:positionalDeterminacy}
    Pareto fronts  are positionally determined, that is,
     for any entrance $i$, the following equality holds: 
     \begin{align*}
        \paretoCurve{\opg{A}}{i}  =  \best{\Big\{ \worst{\big\{ \semPlay{\indPlay{\sigma_{\eve}}{\sigma_{\adam}}{i}}  \  \big | \ \text{a positional $\adam$-strategy }\sigma_{\adam}\big\}} \ \Big |\  \text{a positional $\eve$-strategy }\sigma_{\eve} \Big\}}. 
     \end{align*}
\end{theorem}
Once we establish the positional determinacy of Pareto fronts,
it becomes clear that solving Pareto fronts is decidable through the enumeration of all positional strategies.

We use the following weaker property for proving~\cref{thm:positionalDeterminacy}. 

\begin{lemma}
\label{lem:weakHalfPosDet} 
Let $\opg{A}$ be an oPG, $i$ be an entrance, and $\sigma_{\eve}$ be a positional $\eve$-strategy. 
Then, positional $\adam$-strategies suffice, that is,
the following equality holds: 
\begin{align*}
    \worst{\big\{ \semPlay{\indPlay{\sigma_{\eve}}{\sigma_{\adam}}{i}}  \  \big | \ \sigma_{\adam}\text{ is a $\adam$-strategy}\big\}} = \worst{\big\{ \semPlay{\indPlay{\sigma_{\eve}}{\sigma_{\adam}}{i}}  \  \big | \ \sigma_{\adam}\text{ is a positional $\adam$-strategy}\big\}}.
\end{align*}
\end{lemma}
\begin{proof}
    We take a (finite) oPG $\opg{A}[\sigma_{\eve}]$ that is induced by the positional $\eve$-strategy $\sigma_{\eve}$;
    the nodes and edges of $\opg{A}[\sigma_{\eve}]$ are the same of $\opg{A}$, except the edges from nodes of Player $\eve$ are determined by $\sigma_{\eve}$. 
    Take an optimal $\adam$-strategy $\sigma_{\adam}$ on $\opg{A}[\sigma_{\eve}]$ from $i$, which is also optimal on $\opg{A}$ w.r.t. $\sigma_{\eve}$. 
    If the play $\indPlay{\sigma_{\eve}}{\sigma_{\adam}}{i}$ is losing, then there is a positional optimal $\adam$ strategy by the positional determinacy of PGs. 
    
    Suppose that the play $\indPlay{\sigma_{\eve}}{\sigma_{\adam}}{i}$ reaches an exit $o$. 
    We write the sequence $v_1, v_2, \dots, v_n$ of nodes for the play $\indPlay{\sigma_{\eve}}{\sigma_{\adam}}{i}$, 
    where $n$ is the smallest index such that $v_n = o$. 
    If the equality $v_j = v_j$ implies $i = j$, 
    then we can easily construct a positional $\adam$-strategy $\tau_{\adam}$ such that  $\indPlay{\sigma_{\eve}}{\sigma_{\adam}}{i} =  \indPlay{\sigma_{\eve}}{\tau_{\adam}}{i}$.
    
    The remaining case is that there are $v_i = v_j$ such that $i \not = j$.
    Since $\sigma_{\eve}$ is positional, we can assume that $v_i$ and $v_j$ are owned by Player $\adam$: 
    if every $v_i = v_j$ such that $i \not = j$ is owned by Player $\eve$, it will not reach an exit, which leads to the contradiction. 
    We assume that $i < j$ without loss of generality. 
    If the maximum priority $m$ occurs in $v_i,\dots ,v_j$ is odd, 
    then it contradicts to the fact that $\sigma_{\adam}$ is optimal. 
    If the maximum priority $m$ occurs in $v_i,\dots ,v_j$ is even,
    then the \emph{shortcut} play $p \defeq v_1, \dots, v_i, v_{j+1}, \dots, v_n$ has the same denotation. 
    In fact, if the maximum priority in $\indPlay{\sigma_{\eve}}{\sigma_{\adam}}{i}$ is $m$, 
    then the maximum priority in $p$ is also $m$; otherwise, it contradicts to the fact that $\sigma_{\adam}$ is optimal. 
    If the maximum priority in $\indPlay{\sigma_{\eve}}{\sigma_{\adam}}{i}$ is strictly greater than $m$, then 
    the denotation of $p$ does not change because removing $v_{i+1},\dots, v_{j}$ does not matter. 
    By iterating this shortcut construction, we can finally obtain a play $p' = v'_1, \dots, v'_l$ such that $v'_i = v'_j$ implies $i = j$ and the equality $\semPlay{p'} = \semPlay{\indPlay{\sigma_{\eve}}{\sigma_{\adam}}{i}}$ holds in finite steps. 
    Thus, we can construct a positional $\adam$-strategy $\tau_{\adam}$ such that $\semPlay{\indPlay{\sigma_{\eve}}{\sigma_{\adam}}{i}} =  \semPlay{\indPlay{\sigma_{\eve}}{\tau_{\adam}}{i}}$. 
\end{proof}

\begin{remark}
    The assumption that the $\eve$-strategy $\sigma_{\eve}$ in~\cref{lem:weakHalfPosDet} is positional is indispensable.  
    Consider the oPG shown in~\cref{fig:running_example}, and an $\eve$-strategy  such that the Player $\eve$ repeatedly moves from $c$ to $d$ until she reaches $c$ twice, and then moves to $\exrarg{1}$. 
    Then there is an optimal $\adam$-strategy that moves first from $b$ to $a$, and later from $b$ to $\exlarg{1}$, which is not positional. 
\end{remark}

\subsection{The Loop Construction and the Algorithm}
\label{subsec:loopConst}

In this section, we introduce our loop construction and we prove the positional determinacy of Pareto fronts.  
The parameter of the loop construction is a pair of an entrance and a \emph{query} for oPGs. 
Informally, we show that positional strategies are sufficient to achieve any Pareto-optimal result on the Pareto front, 
and with the loop construction they can be specified by ``minimal'' queries, the number of which is finite.
In addition, we present a simple but efficient algorithm based on the loop construction, which shows improvements over our previous algorithm~\cite{DBLP:journals/corr/abs-2307-08034}.
\begin{definition}[query]
Let $\opg{A}$ be an oPG. A \emph{query}\footnote{We call $q$ query because $q$ asks whether there is an $\eve$-strategy that leads to a better result than the dual of $q$ (see~\cref{lem:correctness_loop_construction}). } $q$ is a function $q\colon \ex^{\opg{A}}\rightarrow \setPriorities\uplus \{\bot\}$. We write $\potentialFuncs{\opg{A}}$ for the set of queries on $\opg{A}$. 
\end{definition}

\begin{figure}[t]
    \centering
    \begin{tikzpicture}[
innodeEve/.style={draw, circle, minimum size=0.5cm},
innodeAdam/.style={draw, diamond, minimum size=0.5cm},
interface/.style={draw, rectangle, minimum size=0.4cm},]
    \fill[lime] (-4cm, -1cm)--(-4cm, 1.2cm)--(-1.5cm, 1.2cm)--(-1.5cm, -1cm)--cycle;
    \node[inner sep=0] (A0enr1d) at (-2.7cm, 0cm) {\scalebox{0.7}{$\cdots$}};
    \node at (-1.7cm, 1cm) {\scalebox{0.7}{$\opg{A}$}};
    \node[inner sep=0] (A0enr1d) at (-4.3cm, 0.5cm) {};
    \node[interface,fill=white, right = 0.5cm of A0enr1d] (A0enr1c) {\scalebox{0.5}{$i_1$}};
    \node[inner sep=0] (A0enr2d) at (-4.3cm, -0.5cm) {};
    \node[interface,fill=white, right = 0.5cm of A0enr2d] (A0enr2c) {\scalebox{0.5}{$i_2$}};
    \node[inner sep=0, right = 0.4cm of A0enr1c] (A0enr1out) {};
    \node[inner sep=0, right = 0.4cm of A0enr2c] (A0enr2out) {};
    \node[inner sep=0, right = 0.8cm of A0enr1c] (A0exr1in) {};
    \node[interface,fill=white, right=0.4cm of A0exr1in] (A0exr1d) {\scalebox{0.5}{$o_1$}};
    \node[inner sep=0, right = 0.4cm of A0exr1d] (A0exr1out) {};
    \node[inner sep=0, right = 0.8cm of A0enr2c] (A0exr2in) {};
    \node[interface,fill=white, right=0.4cm of A0exr2in] (A0exr2d) {\scalebox{0.5}{$o_2$}};
    \node[inner sep=0, right = 0.4cm of A0exr2d] (A0exr2out) {};
    \draw[->, thick] (A0enr1d) -> (A0enr1c);
    \draw[->, thick] (A0enr1c) -> (A0enr1out);
    \draw[->, thick] (A0enr2d) -> (A0enr2c);
    \draw[->, thick] (A0enr2c) -> (A0enr2out);
    \draw[->, thick] (A0enr1c) -> (A0enr1out);
    \draw[->, thick] (A0exr1in) -> (A0exr1d);
    \draw[->, thick] (A0exr2in) -> (A0exr2d);
    \draw[->, thick] (A0exr1d) -> (A0exr1out);
    \draw[->, thick] (A0exr2d) -> (A0exr2out);
    \fill[orange] (0cm, -1cm)--(0cm, 1.2cm)--(2.5cm, 1.2cm)--(2.5cm, -1cm)--cycle;
    \node[inner sep=0] (Aenr1d) at (1.3cm, 0cm) {\scalebox{0.7}{$\cdots$}};
    \node at (1.9cm, 1cm) {\scalebox{0.7}{$\edgeConstructor{\opg{A}}{i_1}{q_1}$}};
    \node[inner sep=0] (Aenr1d) at (-0.3cm, 0.5cm) {};
    \node[interface,fill=white, right = 0.5cm of Aenr1d] (Aenr1c) {\scalebox{0.5}{$i_1$}};
    \node[inner sep=0] (Aenr2d) at (-0.3cm, -0.5cm) {};
    \node[interface,fill=white, right = 0.5cm of Aenr2d] (Aenr2c) {\scalebox{0.5}{$i_2$}};
    \node[inner sep=0, right = 0.4cm of Aenr1c] (Aenr1out) {};
    \node[inner sep=0, right = 0.4cm of Aenr2c] (Aenr2out) {};
    \node[inner sep=0, right = 0.8cm of Aenr1c] (Aexr1in) {};
    \node[interface,fill=white, right=0.4cm of Aexr1in] (Aexr1d) {\scalebox{0.5}{$o_1$}};
    \node[inner sep=0, right = 0.8cm of Aenr2c] (Aexr2in) {};
    \node[interface,fill=white, right=0.4cm of Aexr2in] (Aexr2d) {\scalebox{0.5}{$o_2$}};
    \draw[->, thick] (Aenr1d) -> (Aenr1c);
    \draw[->, thick] (Aenr1c) -> (Aenr1out);
    \draw[->, thick] (Aenr2d) -> (Aenr2c);
    \draw[->, thick] (Aenr2c) -> (Aenr2out);
    \draw[->, thick] (Aenr1c) -> (Aenr1out);
    \draw[->, thick] (Aexr1in) -> (Aexr1d);
    \draw[->, thick] (Aexr2in) -> (Aexr2d);
    \draw[->, thick, color=white] (Aexr1d) to [out=150,in=30] node[yshift=0.1cm,right=-0.5cm] {\scalebox{0.8}{$3$}} (Aenr1c);
    \draw[->, thick, color=white] (Aexr2d) to [loop above] node[yshift=-0.1cm,right=-0.45cm] {\scalebox{0.8}{$1$}} (Aexr2d);
    \fill[cyan] (4cm, -1cm)--(4cm, 1.2cm)--(6.5cm, 1.2cm)--(6.5cm, -1cm)--cycle;
    \node[inner sep=0] (A2enr1d) at (5.3cm, 0cm) {\scalebox{0.7}{$\cdots$}};
    \node at (5.9cm, 1cm) {\scalebox{0.7}{$\edgeConstructor{\opg{A}}{i_1}{q_2}$}};
    \node[inner sep=0] (A2enr1d) at (3.7cm, 0.5cm) {};
    \node[interface,fill=white, right = 0.5cm of A2enr1d] (A2enr1c) {\scalebox{0.5}{$i_1$}};
    \node[inner sep=0] (A2enr2d) at (3.7cm, -0.5cm) {};
    \node[interface,fill=white, right = 0.5cm of A2enr2d] (A2enr2c) {\scalebox{0.5}{$i_2$}};
    \node[inner sep=0, right = 0.4cm of A2enr1c] (A2enr1out) {};
    \node[inner sep=0, right = 0.4cm of A2enr2c] (A2enr2out) {};
    \node[inner sep=0, right = 0.8cm of A2enr1c] (A2exr1in) {};
    \node[interface,fill=white, right=0.4cm of A2exr1in] (A2exr1d) {\scalebox{0.5}{$o_1$}};
    \node[inner sep=0, right = 0.8cm of A2enr2c] (A2exr2in) {};
    \node[interface,fill=white, right=0.4cm of A2exr2in] (A2exr2d) {\scalebox{0.5}{$o_2$}};
    \draw[->, thick] (A2enr1d) -> (A2enr1c);
    \draw[->, thick] (A2enr1c) -> (A2enr1out);
    \draw[->, thick] (A2enr2d) -> (A2enr2c);
    \draw[->, thick] (A2enr2c) -> (A2enr2out);
    \draw[->, thick] (A2enr1c) -> (A2enr1out);
    \draw[->, thick] (A2exr1in) -> (A2exr1d);
    \draw[->, thick] (A2exr2in) -> (A2exr2d);
    \draw[->, thick, color=white] (A2exr1d) to [out=150,in=30] node[yshift=0.1cm,right=-0.5cm] {\scalebox{0.8}{$0$}} (A2enr1c);
    \draw[->, thick, color=white] (A2exr2d) to [out=160,in=290] node[yshift=0.23cm,right=-0.8cm] {\scalebox{0.8}{$2$}} (A2enr1c);
    \end{tikzpicture}
    \caption{The loop constructions $\edgeConstructor{\opg{A}}{i_1}{q_1}$ and $\edgeConstructor{\opg{A}}{i_1}{q_2}$ that add white edges, where queries $q_1, q_2$ are given by (i) $q_1(o_1) \defeq 3$ and $q_1(o_2) \defeq \bot$, and (ii) $q_2(o_1) \defeq 0$ and $q_2(o_2) \defeq 2$,}
    \label{fig:loop_construction}
\end{figure} 

Given an oPG $\opg{A}$, a pair $(i, q)$ of an entrance $i$ and a query $q$ induces a PG $\edgeConstructor{\opg{A}}{i}{q}$ by the loop construction.

\begin{definition}[loop construction]
\label{def:loop_construction}
Let $\opg{A}\defeq (\pg{G}^{\opg{A}}, \interface^{\opg{A}})$ be an oPG, $i$ be an entrance, and $q$ be a query. The \emph{loop construction} $\edgeConstructor{\opg{A}}{i}{q}$ builds a PG $(V^{\opg{A}}, V^{\opg{A}}_{\eve}, V^{\opg{A}}_{\adam}, E, \Omega)$, where $E$ is the least relation and $\Omega$ is the function such that 
\begin{align*}
    &\big((v, v')\in E^{\opg{A}} \text{ and } v\not\in\ex^{\opg{A}}\big) \Longrightarrow \big((v, v')\in E,\text{ and }\ \Omega(v, v') \defeq \Omega^{\opg{A}}(v, v')\big),\\
    & \big(o\in \ex^{\opg{A}} \text{ and } q(o) \in \setPriorities \big)\Longrightarrow \big((o, i)\in E,\text{ and }\ \Omega(o, i) \defeq q(o)\big),\text{ and }\\
    & \big(o\in \ex^{\opg{A}} \text{ and } q(o) = \bot\big) \Longrightarrow\big( (o, o)\in E,\text{ and }\ \Omega(o, o) \defeq 1\big).
\end{align*}
\end{definition}
We present two examples in~\cref{fig:loop_construction}: 
for each exit $o$, the loop construction $\edgeConstructor{\opg{A}}{i}{q}$ adds 
(i) an edge to the specified entrance $i$ whose assigned priority is $q(o)$, or (ii) makes a self loop with the priority $1$ 
if $q(o) = \bot$.

The loop construction is closely related with Pareto fronts of oPGs via \emph{dual} of queries.
Firstly, we extend the sub-priority order $\priorityOrder$ to the functorial order of queries. 
We write $\enrichPriorityOrder$ for the total order on $\setPriorities\uplus \{\top,\bot\}$ such that $m \enrichPriorityOrder m'$ iff (i) $m$ is the least element $\bot$, (ii) $m\priorityOrder m'$, or (iii) $m'$ is the greatest element $\top$. 

\begin{definition}[query order]
    Let $\opg{A}$ be an oPG. The \emph{query order} $\potentialOrder{\opg{A}}$ on $\potentialFuncs{\opg{A}}$ is the functorial order w.r.t. the (restricted) total order $\enrichPriorityOrder$ on $\setPriorities\uplus \{\bot\}$, 
    that is, the inequality $q \potentialOrder{\opg{A}} q'$ holds if $q(o) \enrichPriorityOrder q'(o)$ for each $o\in \ex^{\opg{A}}$. 
\end{definition}

Given a query $q$ and assuming $q(o) = m\in \setPriorities$, the  \emph{dual} $\dualPriority{m}$ of $m$,  
which is also called the left-residual of $0$ in~\cite{DBLP:conf/csl/TsukadaO14}, reveals the relationship between queries and Pareto fronts. 
We define the dual $\dualPriority{m}$ as the lowest priority such that priorities that occur infinitely often are only $m$ and $\dualPriority{m}$, and it satisfies the parity condition (see~\cref{lem:dual_and_maxEven}).   
\begin{definition}[dual]
\label{def:least_passed_priority}
Let $m\in \setPriorities$. The \emph{dual}
$\dualPriority{m}\in \setPriorities$ is given by (i) $\dualPriority{m}\defeq m+1$ if $m$ is odd, (ii) $\dualPriority{m}\defeq 0$ if $m=0$, and (iii) $\dualPriority{m}\defeq m-1$ if $m$ is even and $m > 0$. 
\end{definition}

\begin{lemma}[\!\!\cite{DBLP:conf/csl/TsukadaO14}]
    \label{lem:dual_and_maxEven}
    Let $m, m'\in \setPriorities$. The max priority $\max(m, m')$ is even iff $\dualPriority{m} \priorityOrder m'$.  
\end{lemma}

We define the \emph{dual} of queries based on the dual of priorities. 
\begin{definition}[dual of query]
    Let $\opg{A}$ be an oPG, $i$ be an entrance, and $q$ be a query. 
    The \emph{dual} $\queryAnswer{q}$ of the query $q$ is the result $\queryAnswer{q}\in \setPoints{\opg{A}} $ that is defined as follows: 
    \begin{itemize}
        \item if $q(o) = \bot$ for any $o\in \ex^{\opg{A}}$, then $\queryAnswer{q}\defeq \{\top\}$, and 
        \item otherwise, $\queryAnswer{q}\defeq \{(o, \dualPriority{q(o)})\mid o\in \ex^{\opg{A}}, q(o)\in \setPriorities\}$.
    \end{itemize}
    Note that for any result $\result \in \setPoints{\opg{A}}$ unless $\result = \{\bot\}$, 
    there is a query $q$ such that $\queryAnswer{q} = \result$.
 \end{definition}

\begin{lemma}[duality]
    \label{lem:dualQuery}
    Let $\opg{A}$ be an oPG, $i$ be an entrance, and $q, q'$ be queries. 
    The equality $q = q'$ holds 
    iff the equality $\queryAnswer{q}=\queryAnswer{q'}$ holds. 
    Moreover, the inequality $q \potentialOrder{\opg{A}} q'$ holds 
    iff the inequality $\queryAnswer{q}\succPointOrder{\opg{A}}\queryAnswer{q'}$ holds. 
\end{lemma}
\begin{proof}
    Clearly, the equality $q = q'$ holds iff the equality $\queryAnswer{q}= \queryAnswer{q'}$ holds. 
    
    Suppose that $q  \potentialOrder{\opg{A}} q'$ holds. Then, either $q(o) = \bot$ or $q(o) \priorityOrder q'(o)$ hold for each $o\in O$. 
    For any $(o, m)\in \queryAnswer{q}$, the equality $m = \dualPriority{q(o)} $ holds by definition, and $ \dualPriority{q'(o)} \priorityOrder m $. 
    We thus conclude $\queryAnswer{q}\succPointOrder{\opg{A}}\queryAnswer{q'}$. The converse direction holds in the same way. 
    
\end{proof}
The lemma presented below is crucial for proving the correctness of the loop construction (\cref{prop:minimalQueryPareto}). This, in turn, directly contributes to the proof of the positional determinacy of Pareto fronts.
\begin{lemma}
\label{lem:correctness_loop_construction}
    Let $\opg{A}\defeq (\pg{G}^{\opg{A}}, \interface^{\opg{A}})$ be an oPG, $i$ be an entrance, and $q$ be a query.
    Then, the following conditions are equivalent: 
    \begin{enumerate}
        \item the node $i$ on the PG $\edgeConstructor{\opg{A}}{i}{q}$ is winning, 
        \item there is a result $\result$ such that $\queryAnswer{q} \resultOrder{\opg{A}} \result$ and $\result\in \paretoCurve{\opg{A}}{i}$. 
    \end{enumerate}
\end{lemma}
\begin{proof}
Let $q(o)=\bot$ for any $o\in \ex^{\opg{A}}$. Since $\queryAnswer{q}\defeq \{\top\}$ is the greatest element, it suffices to prove that the following conditions are equivalent: 
\begin{enumerate}
    \item \label{item:winning_cond1}the node $i$ on the PG $\edgeConstructor{\opg{A}}{i}{q}$ is winning,  
    \item \label{item:winning_cond2}$\paretoCurve{\opg{A}}{i} = \big\{\{\top \}\big\}$. 
\end{enumerate}
Suppose~\cref{item:winning_cond1} holds,
then there is an optimal $\eve$-strategy such that every play that is induced by a $\adam$-strategy can not reach any exit $o\in \ex^{\opg{A}}$ because of the loop construction:
every exit has a self-loop with the priority $1$.
Therefore, the optimal $\eve$-strategy is also optimal for $\opg{A}$ from $i$ as well, and~\cref{item:winning_cond2} holds.
The converse direction holds by the same argument.

Let $q(o)\not = \bot$ for some $o\in \ex^{\opg{A}}$. 
Suppose that the node $i$ on the $\edgeConstructor{\opg{A}}{i}{q}$ is winning.
Because of the positional determinacy of PGs, there is a winning positional $\eve$-strategy $\sigma^{\eve}$ on $\edgeConstructor{\opg{A}}{i}{q}$.
By~\cref{lem:weakHalfPosDet}, we only consider positional $\adam$-strategies. 
For any positional $\adam$-strategy $\tau^{\adam}$ on $\opg{A}$, the play $\indPlay{\sigma_{\eve}}{\sigma_{\adam}}{i}$ that is induced by $\sigma_{\eve}$ and $\sigma_{\adam}$ satisfies the following condition:
if $\indPlay{\sigma_{\eve}}{\sigma_{\adam}}{i}$ does not reach any exits $o\in \ex^{\opg{A}}$, then $\indPlay{\sigma_{\eve}}{\sigma_{\adam}}{i}$ is winning.
This is because $\sigma_{\eve}$ is  winning on $\edgeConstructor{\opg{A}}{i}{q}$.
If $\indPlay{\sigma_{\eve}}{\sigma_{\adam}}{i}$ reaches an exit $o\in \ex^{\opg{A}}$, then the maximal priority $m$ that occurs in $\indPlay{\sigma_{\eve}}{\sigma_{\adam}}{i}$ until reaching $o$ satisfies that $\max\big(q(o), m\big)$ is even;
here $q(o)\in \setPriorities$ because $\sigma_{\eve}$ is winning on $\edgeConstructor{\opg{A}}{i}{q}$. 
This is equivalent to $\dualPriority{q(o)} \priorityOrder m$ by~\cref{lem:dual_and_maxEven}. 
Thus, we conclude that $\queryAnswer{q} \resultOrder{\opg{A}}\worst{\big\{ \semPlay{\indPlay{\sigma_{\eve}}{\sigma_{\adam}}{i}}  \  \big | \ \sigma_{\adam}\text{ is a $\adam$-strategy}\big\}}$,
which proves~\cref{item:winning_cond2}.

Conversely, suppose that there is a result $\result$ such that $\queryAnswer{q} \resultOrder{\opg{A}} \result$ and $\result\in \paretoCurve{\opg{A}}{i}$. 
Take an $\eve$-strategy $\sigma_{\eve}$ such that $\result = \worst{\big\{ \semPlay{\indPlay{\sigma_{\eve}}{\sigma_{\adam}}{i}}  \  \big | \ \sigma_{\adam}\text{ is a $\adam$-strategy}\big\}}$;
 note that  $\sigma_{\eve}$ may not be positional. 
Then, the $\eve$-strategy $\sigma_{\eve}$ natually induces an $\eve$-strategy $\sigma'_{\eve}$ on $\edgeConstructor{\opg{A}}{i}{q}$ that mimicks $\sigma_{\eve}$.
It is easy to show that the $\eve$-strategy $\sigma'_{\eve}$ is indeed winning by a similar argument.


\end{proof}

The following proposition is a direct consequence of~\cref{lem:dualQuery,lem:correctness_loop_construction}.
\begin{proposition}
    \label{prop:minimalQueryPareto}
    Let $\opg{A}$ be an oPG, $i$ be an entrance, and $q$ be a query.
    The following conditions are equivalent: 
    \begin{enumerate}
        \item the query $q$ is minimal such that the entrance $i$ on the PG $\edgeConstructor{\opg{A}}{i}{q}$ is winning,
        \item the dual $\queryAnswer{q}$ is Pareto-optimal, that is,  $\queryAnswer{q}\in \paretoCurve{\opg{A}}{i}$. 
    \end{enumerate}
\end{proposition}

Finally, we prove the positional determinacy of Pareto fronts.  
\begin{proof}[Proof of~\cref{thm:positionalDeterminacy}]
    Suppose that the entrance $i$ of $\opg{A}$ is not losing, that is, $\paretoCurve{\opg{A}}{i} \not = \big\{\{\bot\}\big\}$ .
    By~\cref{prop:minimalQueryPareto} (and its proof),
    every Pareto-optimal result $\result$ induces the corresponding minimal query $q$ such that the PG $\edgeConstructor{\opg{A}}{i}{q}$ is winning. 
    Since the PG $\edgeConstructor{\opg{A}}{i}{q}$ is positionally determined, there is a positional winning $\eve$-strategy $\tau_{\eve}$ on $\edgeConstructor{\opg{A}}{i}{q}$; 
    we also regard $\tau_{\eve}$ as a positional $\eve$-strategy on  $\opg{A}$. 
    By~\cref{lem:weakHalfPosDet}, we only consider positional optimal $\adam$-strategies $\tau_{\adam}$ w.r.t. $\tau_{\eve}$.
    If the play $\indPlay{\tau_{\eve}}{\tau_{\adam}}{i}$ does not reach an exit, then the play $\indPlay{\tau_{\eve}}{\tau_{\adam}}{i}$ is winning. 
    If the play $\indPlay{\tau_{\eve}}{\tau_{\adam}}{i}$ reach an exit $o$, then the equality $\dualPriority{q(o)} = m$ hold, 
    where $m$ is the maximum priority that occurs in the play $\indPlay{\tau_{\eve}}{\tau_{\adam}}{i}$; here we use the mimality of $q$ and $\tau_{\adam}$.
    Thus, we can conclude that  $\worst{\big\{ \semPlay{\indPlay{\tau_{\eve}}{\sigma_{\tau}}{i}}  \  \big | \ \tau_{\adam}\text{ is a positional $\adam$-strategy}\big\}} = \queryAnswer{q}$.
    
    The case that the entrance $i$ of $\opg{A}$ is losing is easy to show. This is because of the positional determinacy of parity games.  
\end{proof}

\myparagraph{Algorithm and Time Complexity}
\label{subsec:alg_time_complexity}
The positional determinacy of oPGs that is shown in~\cref{thm:positionalDeterminacy} implies that solving oPGs is decidable by enumerating positional strategies. 
Furthermore,~\cref{prop:minimalQueryPareto} suggests that there is an efficient algorithm based on the loop construction when the number of exists are small. 
Specifically, we evaluate whether the PG $\edgeConstructor{\opg{A}}{i}{q}$ is winning for each query $q$, 
and we keep only those minimal queries that are wininng, whose duals are Pareto-optimal. 
Notably, we can use well-studied algorithms of PGs to solve the PG $\edgeConstructor{\opg{A}}{i}{q}$ without the need to enumerate positional strategies.
\cref{alg:paretoCurveOPGs} shows the details of our algorithm.

\begin{algorithm}[t]
    \caption{\texttt{solveParetoFront}: Solving Pareto fronts for oPGs}
    \begin{algorithmic}[1]
    \State \textbf{Input:} an oPG $\opg{A}$ and an entrance $i$
    \State \textbf{Output:} the Pareto front $\paretoCurve{\opg{A}}{i}$
    \State initialize $S\defeq \big\{\{\bot\}\big\}$  
    \For{each query $q$}
        \State solve the PG $\edgeConstructor{\opg{A}}{i}{q}$
        \If{$\edgeConstructor{\opg{A}}{i}{q}$ is winning}
            \State add $\queryAnswer{q}$ in $S$
        \EndIf
    \EndFor
    \State $S\leftarrow \mathtt{removeSuboptimal}(S)$
    \State \textbf{return} $S$.
    \end{algorithmic}
    \label{alg:paretoCurveOPGs}
\end{algorithm}

\begin{proposition}  
\label{prop:complexityOPG}
 \cref{alg:paretoCurveOPGs} requires time $\mathcal{O}\big((M+2)^N\cdot L + N\cdot (M+2)^{2N}\big)$ to 
 compute the Pareto front $\paretoCurve{\opg{A}}{i}$ of an oPG $\opg{A}$ and an entrance $i$, where $N$ is the number of exits, $M$ is the largest priority,
 and $L$ is the required time for a given algorithm to solve PGs that are built by the loop construction. 
\end{proposition}
\begin{proof}
    Since the number of queries $Q$ is $(M+2)^N$, we solve PGs  $Q$ times, 
    which requires time $\mathcal{O}\big(Q\cdot L\big)$. 
    Removing sub-optimal results can be done in $\mathcal{O}\big(N\cdot Q^2)$ by comparing all results: 
    for comparing two results, we see exits at most $N$ times. 
\end{proof}
Note that~\cref{alg:paretoCurveOPGs} requires an exponential memory to the number of exists since the size of Pareto fronts itself can be exponential. 
A variant of our previous algorithm~\cite{DBLP:journals/corr/abs-2307-08034}, designed for open parity games,
involves enumerating all positional strategies, the number of which is $\mathcal{O}\big((M_{E})^{|V|} \big)$, where 
$M_E$ represents the maximum number of the outgoing edges, and $V$ is the set of nodes. 
In contrast, \cref{alg:paretoCurveOPGs} is exponential only in the number of exits, benefiting from the quasi-polymial-time algorithms for parity games (e.g.~\cite{DBLP:journals/siamcomp/CaludeJKLS22,DBLP:conf/lics/JurdzinskiL17,DBLP:journals/sttt/FearnleyJKSSW19,DBLP:conf/lics/Lehtinen18}).

\section{Solving String Diagrams of Parity Games}
\label{sec:solvingSD}

We apply~\cref{alg:paretoCurveOPGs} that solves oPGs for compositionally solving string diagrams of PGs.  

For compositionally solving a given PG $\semantics{\sd{D}}$, which is defined in~\cref{def:sd_opgs},
 we translate Pareto fronts of oPGs into \emph{shortcut} oPGs. 
 A similar shortcut construction is introduced for MDPs~\cite{WVHRJ2024accepted}, where there is the unique Player. 
\begin{definition}[shortcut oPG]
    \label{def:shortcutoPG}
    Let $\opg{A}\defeq (\pg{G}^{\opg{A}}, \interface^{\opg{A}})$ be an oPG. The \emph{shortcut oPG} $\shortcut{\opg{A}}$ of $\opg{A}$ is the oPG $(\pg{G}, \interface^{\opg{A}})$, where the parity game $\pg{G}\defeq (V, V_{\eve}, V_{\adam}, E, \Omega)$ is defined by 
\begin{itemize}
    \item $V\defeq \interface^{\opg{A}} \uplus \coprod_{i\in \en^{\opg{A}}} \paretoCurve{\opg{A}}{i}$, $V_{\eve} \defeq \interface^{\opg{A}}$, and $V_{\adam} \defeq \coprod_{i\in \en^{\opg{A}}}  \paretoCurve{\opg{A}}{i}$, and
    \item we simultaneously define $E$ and $\Omega$: for any $\result\in  \paretoCurve{\opg{A}}{i}$, $\big(i, \result\big)\in E$ and $\Omega\big(i, \result\big) \defeq 0$, and for any $d\in \result$, 
    \begin{itemize}
        \item if $d = \top$, then $(\result, \result)\in E$ and $\Omega(\result, \result) \defeq 0$, 
        \item if $d = (o, m)$, then $(\result, o\big)\in E$ and $\Omega(\result, o\big) \defeq m$, and
        \item if $d = \bot$, then $(\result, \result)\in E$ and $\Omega(\result, \result) \defeq 1$.
    \end{itemize}
\end{itemize}
\end{definition}

\begin{figure}
    \centering
    \scalebox{1}{
    \begin{tikzpicture}[
innodeEve/.style={draw, circle, minimum size=0.5cm},
innodeAdam/.style={draw, diamond, minimum size=0.5cm},
interface/.style={draw, rectangle, minimum size=0.5cm},]
        \node[interface,fill=white, yshift=1cm] (s1) {\scalebox{0.8}{$\enrarg{1}$}};
        \node[interface,fill=white, yshift=-1cm] (s9) {\scalebox{0.8}{$\exlarg{1}$}};
        \node[inner sep=0,right=-1cm of s9] (exl1) {};
        \node[inner sep=0,right=-1cm of s1] (enr1) {};
        \node[innodeAdam,right=1cm of s1, yshift=0.5cm] (s2) {\scalebox{0.8}{$\result_2$}};
        \node[innodeAdam,right=1cm of s1, yshift=-0.6cm] (s3) {\scalebox{0.8}{$\result_3$}};
        \node[innodeAdam,right=1cm of s1, yshift=-1.7cm] (s4) {\scalebox{0.8}{$\result_1$}};
        \node[interface,fill=white,right=1cm of s2, yshift=-0.5cm] (s7) {\scalebox{0.8}{$\exrarg{1}$}};
        \node[inner sep=0,right=0.3cm of s7] (exr1) {};
        \draw[->] (enr1) -> (s1);
        \draw[->] (s1) -> node [above] {$\scalebox{0.8}{0}$} (s2);
        \draw[->] (s1) -> node [above] {$\scalebox{0.8}{0}$} (s3);
        \draw[->] (s1) -> node [above] {$\scalebox{0.8}{0}$} (s4);
        \draw[->] (s9) -> (exl1);
        \draw[->] (s2) -> node [above] {$\scalebox{0.8}{1}$} (s7);
        \draw[->] (s3) -> node [above] {$\scalebox{0.8}{2}$} (s7);
        \draw[->] (s3) -> node [below] {$\scalebox{0.8}{2}$} (s9);
        \draw[->] (s4) -> node [below] {$\scalebox{0.8}{1}$} (s9);
        \draw[->] (s7) -> (exr1);
    \end{tikzpicture}
    }
    \caption{The shortcut oPG $\shortcut{\opg{A}}$ of an oPG $\opg{A}$
     whose Pareto front $\paretoCurve{\opg{A}}{\enrarg{1}}$ is $\Big\{ \big\{ (\exlarg{1}, 1)\big\},\, \big\{  (\exrarg{1}, 1) \big\},\,
     \big\{ (\exrarg{1}, 2), (\exlarg{1}, 2) \big\}
     \Big\}$. We write $\result_1$, $\result_2$, and $\result_3$ for $\{(\exlarg{1}, 1)\}$, $\{  (\exrarg{1}, 1) \}$, and $\{ (\exrarg{1}, 2), (\exlarg{1}, 2) \}$, respectively. 
     Note that the entrance $\enrarg{1}$ is owned by Player $\eve$.  }
    \label{fig:shortcutOPG}
\end{figure}
We present an example to illustrate this shortcut construction in~\cref{fig:shortcutOPG}. 
Given an oPG $\opg{A}$ and its Pareto fronts $\paretoCurves{\opg{A}}$, 
the shortcut oPG $\shortcut{\opg{A}}$ is a summary of $\opg{A}$, that is, 
it eliminates internal nodes in $\opg{A}$, and keeps only optimal results that are stored in the Pareto fronts $\paretoCurves{\opg{A}}$. 
This shortcut construction allows us to reuse the Pareto fronts for different appearances in a given string diagram with a reduced number of nodes.

Formally, the following  two propositions ensure the correctness of this shortcut construction, that is, the compositionality of Pareto fronts with the shortcut construction. 

\begin{proposition}
\label{prop:base_Pareto_curve}
    Let $\opg{A}$ be an oPG. The Pareto front $\paretoCurve{\opg{A}}{i}$ of $\opg{A}$ coincides with the Pareto front $\paretoCurve{\shortcut{\opg{A}}}{i}$ of the shortcut oPG $\shortcut{\opg{A}}$. 
\end{proposition}
\cref{prop:base_Pareto_curve} is easy to prove by the definition of the shortcut construction. 

\begin{proposition}
\label{prop:comp_Pareto_curve}
    Let $\ast\in \{\seqcomp, \oplus\}$, $\opg{A}\ast \opg{B}$ be an oPG, and $i$ be an entrance.
    The Pareto front $\paretoCurve{\opg{A}\ast \opg{B}}{i}$ of $\opg{A}\ast \opg{B}$ coincides with the Pareto front $\paretoCurve{\shortcut{\opg{A}}\ast \shortcut{\opg{B}}}{i}$ of $\shortcut{\opg{A}}\ast\shortcut{\opg{B}}$. 
\end{proposition}
\begin{proofs}
For $\ast = \oplus$,~\cref{prop:comp_Pareto_curve} is easy to prove
 because there is no interaction between $\opg{A}$ and $\opg{B}$. 
 For  $\ast = \seqcomp$, we focus on ``entrance-dependent'' strategies, and prove the monotonicity of entrance-dependent strategies. 
 Furthermore, we can see that entrance-dependent $\forall$-strategies suffice for entrance-dependent $\exists$-strategies, like~\cref{lem:weakHalfPosDet}. 
 Using these facts and the positional determinacy (\cref{thm:positionalDeterminacy}), we can prove the compositionality of the Pareto front. 
 We present formal definitions and proofs in~\cref{subsec:proofAlgSD}. 
\end{proofs}

 \begin{algorithm}[t]
    \caption{\texttt{solveStringDiagrams}: Solving String Diagrams of PGs}
    \begin{algorithmic}[1]
    \State \textbf{Input:} a string diagram $\sd{D}$
    \State \textbf{Output:} the Pareto fronts $\paretoCurves{\semantics{\sd{D}}}$
    \If{$\sd{D} = \opg{A}$} 
    \State \Return $\mathtt{solveParetoFronts}(\opg{A})$ 
    \Comment run~\cref{alg:paretoCurveOPGs} for each entrance
    \ElsIf{$\sd{D} = \sd{D}_1\ast \sd{D}_2$ for $\ast\in \{\seqcomp, \oplus\}$}
        \State $\mathtt{solveStringDiagrams}(\sd{D}_1)$ and  $\mathtt{solveStringDiagrams}(\sd{D}_2)$
        \Comment{recursion}
        \State construct shortcut oPGs $\shortcut{\semantics{\sd{D}_1}}$ and $\shortcut{\semantics{\sd{D}_2}}$
        \Comment{\cref{def:shortcutoPG}}
        \State \Return $\mathtt{solveParetoFronts}\big(\shortcut{\semantics{\sd{D}_1}}\ast \shortcut{\semantics{\sd{D}_2}}\big)$ 
    \EndIf
    \end{algorithmic}
    \label{alg:sdPGs}
\end{algorithm}

 By~\cref{prop:base_Pareto_curve,prop:comp_Pareto_curve}, we can compositionally solve string diagrams for PGs: we show a pseudocode in~\cref{alg:sdPGs}. 
 For a given string diagram $\sd{D}$, we inductively solve $\sd{D}$: 
 for the base case, that is, $\sd{D}=\opg{A}$, we directly apply~\cref{alg:paretoCurveOPGs} for $\opg{A}$.
 For the step case, that is, $\sd{D}=\sd{D}_1\ast \sd{D}_2$ for $\ast\in \{\seqcomp, \oplus\}$,
 we recursively  solve $\sd{D}_1$ and $\sd{D}_2$, and construct their shortcut oPGs. 
 By~\cref{prop:comp_Pareto_curve}, solving  $\paretoCurves{\shortcut{\sd{D}_1}\ast \shortcut{\sd{D}_2}}$ leads to the Pareto front of $\semantics{\sd{D}}$.

 We remark that for a string diagram $\sd{D}$ such that $\semantics{\sd{D}}$ has no exists,~\cref{alg:sdPGs} runs  in exponential-time, whereas monolithic algorithms---obtained by disregarding the compositional structure---can run in quasi-polynomial time (e.g.~\cite{DBLP:journals/siamcomp/CaludeJKLS22}).
Nevertheless,~\cref{alg:sdPGs} can reuse the solutions of individual components across multiple occurrences, which is particularly beneficial when the system contains duplicate subsystems. 
For instance,~\cite{WVHRJ2024accepted,WatanabeVJH24} exploit such duplication to verify hierarchical systems, where reusability is key to achieving high performance.

\section{Conclusion}
We introduce Pareto fronts of oPGs as a multi-objective solution. 
Positional determinacy holds for our multi-objective solution of oPGs, 
which leads to our novel algorithm for solving open parity games with the loop construction. 
We propose the shortcut construction for open parity games, and we provide a compositional algorithm for string diagrams of parity games by composing Pareto fronts.

A future work is to support the mean-payoff objective, where we have to consider the sum of weights along plays.
Although we expect that positional determinacy also holds for a similar semantics for open mean-payoff (or energy) games, 
the number of objectives may become huge because we believe that we have to consider the sum of weights along plays. 
This is challenging to achieve an efficient compositional algorithm for string diagrams of mean-payoff games. 
For this, we believe that an efficient approximation algorithm for Pareto fronts of open mean-payoff games is needed. 

It would be interesting to investigate whether strategy improvement algorithms for parity games (e.g.~\cite{VogeJ00,LuttenbergerMS20}) can be lifted to the setting of string diagrams of parity games.
 In this context, a strategy improvement algorithm for string diagrams would maintain a strategy throughout the iteration process—one that is not yet winning due to the openness of the system, but still encodes essential information for obtaining the compositional solution.
  Pursuing this direction and comparing our algorithm with such strategy improvement approaches is an exciting future direction. 

Extending the diagrammatic language for parity games~\cite{Piedeleu25} to our semantics is another interesting future direction. 
In our work we use the sub-priority order on priorities, while they use the trivial discrete order on priorities. 
This difference is essential, since positional strategies are sufficient for our semantics, while they are not sufficient for their semantics (of oPGs). 

%
%
%
\bibliography{calco2025}

\begin{thebibliography}{10}

\bibitem{DBLP:conf/lics/Abramsky87}
Samson Abramsky.
\newblock Domain theory in logical form.
\newblock In {\em {LICS}}, pages 47--53. {IEEE} Computer Society, 1987.

\bibitem{baez2014categories}
John~C Baez and Jason Erbele.
\newblock Categories in control.
\newblock {\em arXiv preprint arXiv:1405.6881}, 2014.

\bibitem{DBLP:conf/lics/BaezGMS21}
John~C. Baez, Fabrizio Genovese, Jade Master, and Michael Shulman.
\newblock Categories of nets.
\newblock In {\em {LICS}}, pages 1--13. {IEEE}, 2021.

\bibitem{DBLP:conf/stacs/BerwangerDHK06}
Dietmar Berwanger, Anuj Dawar, Paul Hunter, and Stephan Kreutzer.
\newblock Dag-width and parity games.
\newblock In {\em {STACS}}, volume 3884 of {\em Lecture Notes in Computer Science}, pages 524--536. Springer, 2006.

\bibitem{DBLP:journals/jacm/BonchiGKSZ22}
Filippo Bonchi, Fabio Gadducci, Aleks Kissinger, Pawel Sobocinski, and Fabio Zanasi.
\newblock String diagram rewrite theory {I:} rewriting with {F}robenius structure.
\newblock {\em J. {ACM}}, 69(2):14:1--14:58, 2022.

\bibitem{DBLP:journals/pacmpl/BonchiHPSZ19}
Filippo Bonchi, Joshua Holland, Robin Piedeleu, Pawel Sobocinski, and Fabio Zanasi.
\newblock Diagrammatic algebra: from linear to concurrent systems.
\newblock {\em Proc. {ACM} Program. Lang.}, 3({POPL}):25:1--25:28, 2019.

\bibitem{DBLP:conf/popl/BonchiSZ15}
Filippo Bonchi, Pawel Sobocinski, and Fabio Zanasi.
\newblock Full abstraction for signal flow graphs.
\newblock In {\em {POPL}}, pages 515--526. {ACM}, 2015.

\bibitem{DBLP:journals/iandc/BonchiSZ17}
Filippo Bonchi, Pawel Sobocinski, and Fabio Zanasi.
\newblock The calculus of signal flow diagrams {I:} linear relations on streams.
\newblock {\em Inf. Comput.}, 252:2--29, 2017.

\bibitem{DBLP:conf/concur/BruyereRT22}
V{\'{e}}ronique Bruy{\`{e}}re, Jean{-}Fran{\c{c}}ois Raskin, and Cl{\'{e}}ment Tamines.
\newblock Pareto-rational verification.
\newblock In {\em {CONCUR}}, volume 243 of {\em LIPIcs}, pages 33:1--33:20. Schloss Dagstuhl - Leibniz-Zentrum f{\"{u}}r Informatik, 2022.

\bibitem{DBLP:journals/siamcomp/CaludeJKLS22}
Cristian~S. Calude, Sanjay Jain, Bakhadyr Khoussainov, Wei Li, and Frank Stephan.
\newblock Deciding parity games in quasi-polynomial time.
\newblock {\em {SIAM} J. Comput.}, 51(2):17--152, 2022.

\bibitem{DBLP:conf/mfcs/ChenFKSW13}
Taolue Chen, Vojtech Forejt, Marta~Z. Kwiatkowska, Aistis Simaitis, and Clemens Wiltsche.
\newblock On stochastic games with multiple objectives.
\newblock In {\em {MFCS}}, volume 8087 of {\em Lecture Notes in Computer Science}, pages 266--277. Springer, 2013.

\bibitem{DBLP:conf/lics/ClarkeLM89}
Edmund~M. Clarke, David~E. Long, and Kenneth~L. McMillan.
\newblock Compositional model checking.
\newblock In {\em {LICS}}, pages 353--362. {IEEE} Computer Society, 1989.

\bibitem{DBLP:conf/icalp/CoeckeD08}
Bob Coecke and Ross Duncan.
\newblock Interacting quantum observables.
\newblock In {\em {ICALP} {(2)}}, volume 5126 of {\em Lecture Notes in Computer Science}, pages 298--310. Springer, 2008.

\bibitem{DBLP:books/cu/CK2017}
Bob Coecke and Aleks Kissinger.
\newblock {\em Picturing Quantum Processes: {A} First Course in Quantum Theory and Diagrammatic Reasoning}.
\newblock Cambridge University Press, 2017.

\bibitem{DBLP:conf/lop/Emerson85}
E.~Allen Emerson.
\newblock Automata, tableaux and temporal logics (extended abstract).
\newblock In {\em Logic of Programs}, volume 193 of {\em Lecture Notes in Computer Science}, pages 79--88. Springer, 1985.

\bibitem{DBLP:conf/focs/EmersonJ91}
E.~Allen Emerson and Charanjit~S. Jutla.
\newblock Tree automata, mu-calculus and determinacy (extended abstract).
\newblock In {\em {FOCS}}, pages 368--377. {IEEE} Computer Society, 1991.

\bibitem{DBLP:journals/tcs/EmersonJS01}
E.~Allen Emerson, Charanjit~S. Jutla, and A.~Prasad Sistla.
\newblock On model checking for the {\(\mathrm{\mu}\)}-calculus and its fragments.
\newblock {\em Theor. Comput. Sci.}, 258(1-2):491--522, 2001.

\bibitem{DBLP:journals/lmcs/EtessamiKVY08}
Kousha Etessami, Marta~Z. Kwiatkowska, Moshe~Y. Vardi, and Mihalis Yannakakis.
\newblock Multi-objective model checking of {M}arkov decision processes.
\newblock {\em Log. Methods Comput. Sci.}, 4(4), 2008.

\bibitem{DBLP:journals/sttt/FearnleyJKSSW19}
John Fearnley, Sanjay Jain, Bart de~Keijzer, Sven Schewe, Frank Stephan, and Dominik Wojtczak.
\newblock An ordered approach to solving parity games in quasi-polynomial time and quasi-linear space.
\newblock {\em Int. J. Softw. Tools Technol. Transf.}, 21(3):325--349, 2019.

\bibitem{DBLP:journals/corr/abs-1112-0221}
John Fearnley and Sven Schewe.
\newblock Time and space results for parity games with bounded treewidth.
\newblock {\em Log. Methods Comput. Sci.}, 9(2), 2013.

\bibitem{DBLP:conf/atva/ForejtKP12}
Vojtech Forejt, Marta~Z. Kwiatkowska, and David Parker.
\newblock Pareto curves for probabilistic model checking.
\newblock In {\em {ATVA}}, volume 7561 of {\em Lecture Notes in Computer Science}, pages 317--332. Springer, 2012.

\bibitem{DBLP:conf/aplas/FujimaIK13}
Koichi Fujima, Sohei Ito, and Naoki Kobayashi.
\newblock Practical alternating parity tree automata model checking of higher-order recursion schemes.
\newblock In {\em {APLAS}}, volume 8301 of {\em Lecture Notes in Computer Science}, pages 17--32. Springer, 2013.

\bibitem{DBLP:conf/lics/GhaniHWZ18}
Neil Ghani, Jules Hedges, Viktor Winschel, and Philipp Zahn.
\newblock Compositional game theory.
\newblock In {\em {LICS}}, pages 472--481. {ACM}, 2018.

\bibitem{DBLP:conf/mfcs/GrelloisM15}
Charles Grellois and Paul{-}Andr{\'{e}} Melli{\`{e}}s.
\newblock Finitary semantics of linear logic and higher-order model-checking.
\newblock In {\em {MFCS} {(1)}}, volume 9234 of {\em Lecture Notes in Computer Science}, pages 256--268. Springer, 2015.

\bibitem{DBLP:conf/lics/JurdzinskiL17}
Marcin Jurdzinski and Ranko Lazic.
\newblock Succinct progress measures for solving parity games.
\newblock In {\em {LICS}}, pages 1--9. {IEEE} Computer Society, 2017.

\bibitem{DBLP:conf/cav/Kaivola92}
Roope Kaivola.
\newblock Compositional model checking for linear-time temporal logic.
\newblock In {\em {CAV}}, volume 663 of {\em Lecture Notes in Computer Science}, pages 248--259. Springer, 1992.

\bibitem{DBLP:conf/lics/KobayashiO09}
Naoki Kobayashi and C.{-}H.~Luke Ong.
\newblock A type system equivalent to the modal mu-calculus model checking of higher-order recursion schemes.
\newblock In {\em {LICS}}, pages 179--188. {IEEE} Computer Society, 2009.

\bibitem{DBLP:journals/tcs/Kozen83}
Dexter Kozen.
\newblock Results on the propositional mu-calculus.
\newblock {\em Theor. Comput. Sci.}, 27:333--354, 1983.

\bibitem{DBLP:conf/stoc/KupfermanV98}
Orna Kupferman and Moshe~Y. Vardi.
\newblock Weak alternating automata and tree automata emptiness.
\newblock In {\em {STOC}}, pages 224--233. {ACM}, 1998.

\bibitem{DBLP:journals/iandc/KwiatkowskaNPQ13}
Marta~Z. Kwiatkowska, Gethin Norman, David Parker, and Hongyang Qu.
\newblock Compositional probabilistic verification through multi-objective model checking.
\newblock {\em Inf. Comput.}, 232:38--65, 2013.

\bibitem{DBLP:conf/csl/LavoreH021}
Elena~Di Lavore, Jules Hedges, and Pawel Sobocinski.
\newblock Compositional modelling of network games.
\newblock In {\em {CSL}}, volume 183 of {\em LIPIcs}, pages 30:1--30:24. Schloss Dagstuhl - Leibniz-Zentrum f{\"{u}}r Informatik, 2021.

\bibitem{DBLP:conf/lics/Lehtinen18}
Karoliina Lehtinen.
\newblock A modal {\(\mu\)} perspective on solving parity games in quasi-polynomial time.
\newblock In {\em {LICS}}, pages 639--648. {ACM}, 2018.

\bibitem{LuttenbergerMS20}
Michael Luttenberger, Philipp~J. Meyer, and Salomon Sickert.
\newblock Practical synthesis of reactive systems from {LTL} specifications via parity games.
\newblock {\em Acta Informatica}, 57(1-2):3--36, 2020.

\bibitem{DBLP:conf/cav/Obdrzalek03}
Jan Obdrz{\'{a}}lek.
\newblock Fast mu-calculus model checking when tree-width is bounded.
\newblock In {\em {CAV}}, volume 2725 of {\em Lecture Notes in Computer Science}, pages 80--92. Springer, 2003.

\bibitem{DBLP:conf/lics/Ong06}
C.{-}H.~Luke Ong.
\newblock On model-checking trees generated by higher-order recursion schemes.
\newblock In {\em {LICS}}, pages 81--90. {IEEE} Computer Society, 2006.

\bibitem{DBLP:conf/focs/PapadimitriouY00}
Christos~H. Papadimitriou and Mihalis Yannakakis.
\newblock On the approximability of trade-offs and optimal access of web sources.
\newblock In {\em {FOCS}}, pages 86--92. {IEEE} Computer Society, 2000.

\bibitem{Piedeleu25}
Robin Piedeleu.
\newblock The algebra of parity games.
\newblock {\em CoRR}, abs/2501.18499, 2025.

\bibitem{DBLP:conf/rp/RathkeSS14}
Julian Rathke, Pawel Sobocinski, and Owen Stephens.
\newblock Compositional reachability in {P}etri nets.
\newblock In {\em {RP}}, volume 8762 of {\em Lecture Notes in Computer Science}, pages 230--243. Springer, 2014.

\bibitem{DBLP:journals/iandc/SalvatiW14}
Sylvain Salvati and Igor Walukiewicz.
\newblock Krivine machines and higher-order schemes.
\newblock {\em Inf. Comput.}, 239:340--355, 2014.

\bibitem{DBLP:conf/csl/TsukadaO14}
Takeshi Tsukada and C.{-}H.~Luke Ong.
\newblock Compositional higher-order model checking via \emph{{\(\omega\)}}-regular games over {B}{\"{o}}hm trees.
\newblock In {\em {CSL-LICS}}, pages 78:1--78:10. {ACM}, 2014.

\bibitem{DBLP:journals/iandc/VelnerC0HRR15}
Yaron Velner, Krishnendu Chatterjee, Laurent Doyen, Thomas~A. Henzinger, Alexander~Moshe Rabinovich, and Jean{-}Fran{\c{c}}ois Raskin.
\newblock The complexity of multi-mean-payoff and multi-energy games.
\newblock {\em Inf. Comput.}, 241:177--196, 2015.

\bibitem{VogeJ00}
Jens V{\"{o}}ge and Marcin Jurdzinski.
\newblock A discrete strategy improvement algorithm for solving parity games.
\newblock In {\em {CAV}}, volume 1855 of {\em Lecture Notes in Computer Science}, pages 202--215. Springer, 2000.

\bibitem{DBLP:journals/corr/abs-2112-14058}
Kazuki Watanabe, Clovis Eberhart, Kazuyuki Asada, and Ichiro Hasuo.
\newblock A compositional approach to parity games.
\newblock In {\em {MFPS}}, volume 351 of {\em {EPTCS}}, pages 278--295, 2021.

\bibitem{DBLP:conf/cav/WatanabeEAH23}
Kazuki Watanabe, Clovis Eberhart, Kazuyuki Asada, and Ichiro Hasuo.
\newblock Compositional probabilistic model checking with string diagrams of {MDP}s.
\newblock In {\em {CAV} {(3)}}, volume 13966 of {\em Lecture Notes in Computer Science}, pages 40--61. Springer, 2023.

\bibitem{DBLP:journals/corr/abs-2307-08034}
Kazuki Watanabe, Clovis Eberhart, Kazuyuki Asada, and Ichiro Hasuo.
\newblock Compositional solution of mean payoff games by string diagrams.
\newblock In {\em Principles of Verification {(3)}}, volume 15262 of {\em Lecture Notes in Computer Science}, pages 423--445. Springer, 2024.

\bibitem{WVHRJ2024accepted}
Kazuki Watanabe, Marck van~der Vegt, Ichiro Hasuo, Jurriaan Rot, and Sebastian Junges.
\newblock Pareto curves for compositionally model checking string diagrams of {MDP}s.
\newblock In {\em {TACAS} {(2)}}, volume 14571 of {\em Lecture Notes in Computer Science}, pages 279--298. Springer, 2024.

\bibitem{WatanabeVJH24}
Kazuki Watanabe, Marck van~der Vegt, Sebastian Junges, and Ichiro Hasuo.
\newblock Compositional value iteration with pareto caching.
\newblock In {\em {CAV} {(3)}}, volume 14683 of {\em Lecture Notes in Computer Science}, pages 467--491. Springer, 2024.

\end{thebibliography}

\clearpage

\appendix

\section{Proof of~\cref{prop:comp_Pareto_curve} for $\ast = \seqcomp$}
\label{subsec:proofAlgSD}
In this section, we additionally assume that every entrance on a given oPG is not reachable except as an initial node, that is, $(v, v')\in E$ implies $v'\not \in \en$ without loss of generality.  
We fix oPGs $\opg{A}, \opg{B}$, and an entrance $i$ on $\opg{A}\seqcomp \opg{B}$.
Here, we only consider plays that start from the specified entrance $i\in \opg{A}$; note that the sequential composition is bidirectional and the ordering between $\opg{A}$ and $\opg{B}$ is not essential.  
Firstly, we introduce \emph{entrance-dependent} strategies on $\opg{A}\seqcomp \opg{B}$: 
intuitively, for each play, they behave positionally depending on the last entrance that appears in the play.

\begin{definition}[entrance-dependent strategy]
   Let $\sigma_{\eve}$ be an $\eve$-strategy on $\opg{A}\seqcomp \opg{B}$. 
   We write $\posStrategies{\opg{A}}{\eve}$ and $\posStrategies{\opg{B}}{\eve}$ for the set of positional $\eve$-strategies on $\opg{A}$ and  $\opg{B}$, respectively. 
   The $\eve$-strategy $\sigma_{\eve}$ is \emph{entrance-dependent} if 
   there are assignments $\sigma^{\opg{A}}_{\eve, \_}\colon \en^{\opg{A}}\rightarrow \posStrategies{\opg{A}}{\eve}$
   and $\sigma^{\opg{B}}_{\eve, \_}\colon  \en^{\opg{B}}\rightarrow \posStrategies{\opg{B}}{\eve}$
   such that for any play $p = p_1\cdots p_n\in V^{\ast}\times V_{\eve}$, if the last entrance $i'$ that appears in $p$ is $i'\in \en^{\opg{A}}$,  
   then the $\eve$-strategy $\sigma_{\eve}$ is given by $\sigma_{\eve}(p) = \sigma^{\opg{A}}_{\eve, i'}(p_n)$, 
   and same for the case of $i'\in \en^{\opg{B}}$.  

   We also define \emph{entrance-dependent} $\adam$-strategies in the same manner. 
\end{definition}

Entrance-dependent strategies consist of positional strategies and \emph{monotonicity} of positional strategies holds.
Intuitively, if we use better (or worse) positional strategies for entrance-dependent strategies, then their entrance-dependent strategies are also better (or worse).

\begin{lemma}[$\adam$-monotonicity]
   \label{lem:monoAdam}
   Let $\sigma_{\eve}$ be an entrance-dependent $\eve$-strategy and $\sigma_{\adam}, \tau_{\adam}$ be entrance-dependent $\adam$-strategies. 
   Suppose that for each $i'\in  \en^{\opg{A}}$,  
   the inequality 
   $\semPlay{\indPlay{\sigma^{\opg{A}}_{\eve, i'}}{\tau^{\opg{A}}_{\adam, i'}}{i'}} \domainOrder \semPlay{\indPlay{\sigma^{\opg{A}}_{\eve, i'}}{\sigma^{\opg{A}}_{\adam, i'}}{i'}}$ holds, and same for $i'\in\en^{\opg{B}}$.
   Then, the inequality $\semPlay{\indPlay{\sigma_{\eve}}{\tau_{\adam}}{i}}\domainOrder \semPlay{\indPlay{\sigma_{\eve}}{\sigma_{\adam}}{i}}$ holds. 
\end{lemma}
\begin{proof}
    We prove the inequality $\semPlay{\indPlay{\sigma_{\eve}}{\tau_{\adam}}{i}}\domainOrder \semPlay{\indPlay{\sigma_{\eve}}{\sigma_{\adam}}{i}}$. 
    If the play $\indPlay{\sigma_{\eve}}{\sigma_{\adam}}{i}$ is winning, then obviously the inequality holds. 
    
    Suppose that the play $\indPlay{\sigma_{\eve}}{\sigma_{\adam}}{i}$ reaches an exit $o$ (of $\opg{A}\seqcomp \opg{B}$) with the maximum priority $m$, that is, $\semPlay{\indPlay{\sigma_{\eve}}{\sigma_{\adam}}{i}} = (o, m)$. 
    By the assumption and the monotonicity of priorities (\cref{lem:monoPriority}), the play $\indPlay{\sigma_{\eve}}{\tau_{\adam}}{i}$ is either losing, or reaching the same exit with a potentially different maximum priority $m'$ such that $m'\priorityOrder m$.
    Thus, the desired inequality also holds. 

    Suppose that the play $\indPlay{\sigma_{\eve}}{\sigma_{\adam}}{i}$ is losing.
    By collecting the denotations of $\semPlay{\indPlay{\sigma^{\opg{A}}_{\eve, i'}}{\sigma^{\opg{A}}_{\adam, i'}}{i'}}$ and $\semPlay{\indPlay{\sigma^{\opg{B}}_{\eve, i'}}{\sigma^{\opg{B}}_{\adam, i'}}{i'}}$, 
    we obtain an infinte sequence $m_1\cdot m_2\cdots $ of priorities for $\indPlay{\sigma_{\eve}}{\sigma_{\adam}}{i}$ such that 
    the maximum priorities that occurs infinitely often in this sequence is odd. 
    By the assumption, the corresponding infinite sequence $m'_1\cdot m'_2\cdots$ of priorities for $\indPlay{\sigma_{\eve}}{\tau_{\adam}}{i}$ satisfies that 
    $m'_j\priorityOrder m_j$ for each $j\in \nat$. 
    For each $j\in \nat$, if $m_j$ is the maximum priority $m$ that occurs infinitely often in the sequence  $m_1\cdot m_2\cdots $, then $m'_j$ is a bigger or the same odd number, and at least 
    one of the such odd number occurs infinitely often in the sequence $m'_1\cdot m'_2\cdots$ because the size of priorities is finite.  
    If $m_j$ is even, then $m'_j$ is a smaller or the same even number, therefore, the maximum even 
    number that occurs infinitely often in the sequence $m'_1\cdot m'_2\cdots$ is strictly smaller than $m$. 
    Therefore, the maximum number that occurs infinitely often in $m'_1\cdot m'_2\cdots$ is again odd, which means that the desired inequality holds.   
\end{proof}

\begin{lemma}[$\eve$-monotonicity]
   \label{lem:monoEve}
   Let $\sigma_{\eve},\tau_{\eve}$ be entrance-dependent $\eve$-strategies.  
   Suppose that for each $i'\in \en^{\opg{A}}$,   
   the following inequality holds: 
   \begin{align*}
       \worst{\big\{ \semPlay{\indPlay{\tau^{\opg{A}}_{\eve, i'}}{\tau_{\adam}}{i'}}  \  \big | \ \tau^{\opg{A}}_{\adam}\text{ is a $\adam$-strategy}\big\}}\resultOrder{\opg{A}} \worst{\big\{ \semPlay{\indPlay{\sigma^{\opg{A}}_{\eve, i'}}{\sigma_{\adam}}{i'}}  \  \big | \ \sigma^{\opg{A}}_{\adam}\text{ is a $\adam$-strategy}\big\}},
   \end{align*}
   and same for $i'\in \en^{\opg{B}}$ .
   Then, for any entrance-dependent $\adam$-strategy $\sigma_{\adam}$, there is an entrance-dependent $\adam$-strategy $\tau_{\adam}$ such that
        $\semPlay{\indPlay{\tau_{\eve}}{\tau_{\adam}}{i}}   \domainOrder \semPlay{\indPlay{\sigma_{\eve}}{\sigma_{\adam}}{i}}$.
\end{lemma}
\begin{proof}
    For any entrance-dependent $\adam$-strategy $\sigma_{\adam}$, 
    we construct an entrance-dependent $\adam$-strategy $\tau_{\adam}$ such that 
    the inequality  $\semPlay{\indPlay{\tau_{\eve}}{\tau_{\adam}}{i}}   \domainOrder \semPlay{\indPlay{\sigma_{\eve}}{\sigma_{\adam}}{i}}$ holds. 
    For each entrance $i'\in \en^{\opg{A}}$, there is a positional optimal $\adam$-strategy $\sigma'^{\opg{A}}_{\adam, i'}$ such that
    the inequality $\semPlay{\indPlay{\sigma^{\opg{A}}_{\eve, i'}}{\sigma'^{\opg{A}}_{\adam, i'}}{i'}}\domainOrder \semPlay{\indPlay{\sigma^{\opg{A}}_{\eve, i'}}{\sigma^{\opg{A}}_{\adam, i'}}{i'}}$ holds. 
    By the assumption, we can take a positional optimal $\adam$-strategy $\tau^{\opg{A}}_{\adam, i'}$ such that  $\semPlay{\indPlay{\tau^{\opg{A}}_{\eve,i'}}{\tau^{\opg{A}}_{\adam, i'}}{i'}}\domainOrder \semPlay{\indPlay{\sigma^{\opg{A}}_{\eve, i'}}{\sigma'^{\opg{A}}_{\adam, i'}}{i'}}$, 
    thus the inequality  $\semPlay{\indPlay{\tau^{\opg{A}}_{\eve, i'}}{\tau^{\opg{A}}_{\adam, i'}}{i'}}\domainOrder \semPlay{\indPlay{\sigma^{\opg{A}}_{\eve, i'}}{\sigma^{\opg{A}}_{\adam, i'}}{i'}}$ holds. 
    Similarly, we can take a positional optimal $\adam$-strategy $\tau^{\opg{B}}_{\adam, i'}$ for each entrance $i'\in  \en^{\opg{B}}$ 
    such that the inequality $\semPlay{\indPlay{\tau^{\opg{B}}_{\eve, i'}}{\tau^{\opg{B}}_{\adam, i'}}{i'}}\domainOrder \semPlay{\indPlay{\sigma^{\opg{B}}_{\eve, i'}}{\sigma^{\opg{B}}_{\adam, i'}}{i'}}$ holds. 
    By the almost same argument used in the proof of~\cref{lem:monoAdam}, 
    we can conclude that the inequality $\semPlay{\indPlay{\tau_{\eve}}{\tau_{\adam}}{i}}   \domainOrder \semPlay{\indPlay{\sigma_{\eve}}{\sigma_{\adam}}{i}}$ holds, 
    where $\tau_{\adam}$ is the entrance-dependent $\adam$-strategy that consists of these positional optimal $\adam$-strategies $\tau^{\opg{A}}_{\adam, i'}$ and $\tau^{\opg{B}}_{\adam, i'}$. 
\end{proof}

The last ingredient for proving~\cref{prop:comp_Pareto_curve} is the fact that entrance-dependent $\adam$-strategies suffice for entrance-dependent $\eve$-strategies.

\begin{lemma}
   \label{lem:weakHalfEntEepDet} 
   Let $\sigma_{\eve}$ be an entrance-dependent $\eve$-strategy. 
   Then, entrance-dependent $\adam$-strategies suffice, that is,
   the following equality holds: 
   \begin{align*}
       \worst{\big\{ \semPlay{\indPlay{\sigma_{\eve}}{\sigma_{\adam}}{i}}  \  \big | \ \text{a $\adam$-strategy }\sigma_{\adam}\big\}}  = \worst{\big\{ \semPlay{\indPlay{\sigma_{\eve}}{\tau_{\adam}}{i}}  \  \big | \ \text{an entrance-dependent $\adam$-strategy }\tau_{\adam} \big\}}.
   \end{align*}
\end{lemma}
\begin{proof}
    Let $\sigma_{\adam}$ be an optimal $\adam$-strategy. 
    If the play $\indPlay{\sigma_{\eve}}{\sigma_{\adam}}{i}$ is losing, it is straightforward 
    to construct a losing entrance-dependent $\adam$-strategy by distinguishing last entrances, and positional determinacy of PGs.  
    
    Suppose that the play $\indPlay{\sigma_{\eve}}{\sigma_{\adam}}{i}$ reaches an exit $o$ on $\opg{A}\seqcomp \opg{B}$. 
    Let $i_1\cdots i_n$ be the (finite) sequence of entrances in $\en^{\opg{A}} \uplus \en^{\opg{B}}$ that occurs in $\indPlay{\sigma_{\eve}}{\sigma_{\adam}}{i}$, where the ordering is the order of their appearance. 
    Then, we can take a positioal $\adam$-strategy $\tau^{\opg{A}}_{\adam,j}$ for each entrances $i_j$ such that 
    the constructed $\adam$-strategy $\tau_{\adam}$ from these $\tau^{\opg{A}}_{\adam,j}$ satisfies that $\semPlay{\indPlay{\sigma_{\eve}}{\tau_{\adam}}{i}}\domainOrder\semPlay{\indPlay{\sigma_{\eve}}{\sigma_{\adam}}{i}} $. 
    Since $\sigma_{\adam}$ is optimal, the equality $\semPlay{\indPlay{\sigma_{\eve}}{\tau_{\adam}}{i}}=\semPlay{\indPlay{\sigma_{\eve}}{\sigma_{\adam}}{i}} $ holds. 
    
    Here, $\tau_{\adam}$ may not be an entrance-dependent, that is, there are $j\not =  j'$ such that $i_j = i_{j'}$ and  $\tau^{\opg{A}}_{\adam,j} \not = \tau^{\opg{A}}_{\adam,j'}$ in general. 
    If $\tau_{\adam}$ is not an entrance-dependent, we construct an $\adam$-entrance-dependent strategy by a similar argument of the proof of~\cref{lem:weakHalfPosDet}; note that we assume $\sigma_{\eve}$ is entrance-dependent. 
\end{proof}

\begin{definition}[lower preorder]
    Let $\mathbf{T}$ be a set such that $\mathbf{T}\defeq \{\best{T} \mid T\subseteq \setPoints{\opg{A}\seqcomp \opg{B}}, \text{and $T\not = \emptyset$}\}$.
    Then, the \emph{lower preorder} $\lowerOrder$ on $\mathbf{T}$ is defined as follows: $\best{T_1}\lowerOrder \best{T_2}$ if for any result $\result_1\in \best{T_1}$, there is a result $\result_2\in \best{T_2}$
    such that $\result_1\resultOrder{\opg{A}\seqcomp \opg{B}} \result_2$. 
    
\end{definition}

\begin{lemma}
    \label{lem:lowerOrder}
    The preorder $\lowerOrder$ is a partial order. 
\end{lemma}
\begin{proof}
    It suffices to prove that $\best{T_1}\lowerOrder\best{T_2}$ and $\best{T_2}\lowerOrder\best{T_1}$ imply that $\best{T_1} \subseteq \best{T_2}$. 
    For any result $\result_1\in \best{T_1}$, there is a result $\result_2\in \best{T_2}$ such that $\result_1\resultOrder{\opg{A}\seqcomp \opg{B}} \result_2$. 
    Similarly, there is a $\result_3\in \best{T_1}$ such that $\result_2\resultOrder{\opg{A}\seqcomp \opg{B}}\result_3$. 
    Thus, we can conclude that $\result_1 = \result_3$ because of the maximality of $\best{T_1}$. 
    Since the result order is a partial order, we can conclude that $\result_1 = \result_2$, thus the desired $\best{T_1} \subseteq \best{T_2}$ holds. 

\end{proof}

We conclude this section by showing our proof of~\cref{prop:comp_Pareto_curve}.  

\begin{proof}[Proof of~\cref{prop:comp_Pareto_curve} for $\ast = \seqcomp$]
   By definition and~\cref{thm:positionalDeterminacy}, the Pareto front $\paretoCurve{\shortcut{\opg{A}}\seqcomp \shortcut{\opg{B}}}{i}$ is given by 
   \begin{align*}
       \best{\Big\{\worst{\big\{\semPlay{\indPlay{\sigma_{\eve}}{\sigma_{\adam}}{i}} \ \big | \ \text{a positional $\sigma_{\adam}$ on $\shortcut{\opg{A}}\seqcomp \shortcut{\opg{B}}$} \big\}} \ \Big | \ \text{a positional $\sigma_{\eve}$ on $\shortcut{\opg{A}}\seqcomp \shortcut{\opg{B}}$} \Big\}}.
   \end{align*}
   For a positional $\eve$-strategy $\sigma_{\eve}$ on $\shortcut{\opg{A}}\seqcomp \shortcut{\opg{B}}$, there is a corresponding entrance-dependent $\eve$-strategy $\tilde{\sigma_{\eve}}$ on $\opg{A}\seqcomp \opg{B}$ by mimicking $\sigma_{\eve}$.
   Similarly, the following equality holds:  
   \begin{align*}
       &\worst{\big\{\semPlay{\indPlay{\sigma_{\eve}}{\sigma_{\adam}}{i}} \ \big | \ \text{a positional $\sigma_{\adam}$ on $\shortcut{\opg{A}}\seqcomp \shortcut{\opg{B}}$} \big\}}\\
        = &\worst{\big\{\semPlay{\indPlay{\tilde{\sigma_{\eve}}}{\tilde{\sigma_{\adam}}}{i}} \ \big | \ \text{a positional $\sigma_{\adam}$ on $\shortcut{\opg{A}}\seqcomp \shortcut{\opg{B}}$} \big\}},
   \end{align*}
   where the $\adam$-strategy $\tilde{\sigma_{\adam}}$ is the corresponding entrance-dependent $\adam$-strategy of a positional $\adam$-strategy $\sigma_{\adam}$, and 
   the play $\indPlay{\tilde{\sigma_{\eve}}}{\tilde{\sigma_{\adam}}}{i}$ is on $\opg{A}\seqcomp \opg{B}$. 
   By~\cref{lem:monoAdam,lem:weakHalfEntEepDet}, we can further say that the Pareto front $\paretoCurve{\shortcut{\opg{A}}\seqcomp \shortcut{\opg{B}}}{i}$ is given by
   \begin{align*}
       \best{\Big\{\worst{\big\{\semPlay{\indPlay{\tilde{\sigma_{\eve}}}{\sigma_{\adam}}{i}} \ \big | \ \text{a $\adam$-strategy $\sigma_{\adam}$ on $\opg{A}\seqcomp \opg{B}$} \big\}} \ \Big | \ \text{a positional $\sigma_{\eve}$ on $\shortcut{\opg{A}}\seqcomp \shortcut{\opg{B}}$} \Big\}},
   \end{align*}
   By using the lower preorder $\lowerOrder$ on the sets of maximal results, which is a partial order (\cref{lem:lowerOrder}), we obtain the following inequality: 
   \begin{align*}
       &\paretoCurve{\shortcut{\opg{A}}\seqcomp \shortcut{\opg{B}}}{i}\\
       \lowerOrder  &\best{\Big\{\worst{\big\{\semPlay{\indPlay{\sigma_{\eve}}{\sigma_{\adam}}{i}} \ \big | \ \text{a $\adam$-strategy $\sigma_{\adam}$ on $\opg{A}\seqcomp \opg{B}$} \big\}} \ \Big | \ \text{an $\eve$-strategy $\sigma_{\eve}$ on $\opg{A}\seqcomp \opg{B}$} \Big\}}.
   \end{align*}
   On the other hand, since positional strategies are entrance-dependent, we obtain the following inequality by~\cref{lem:monoAdam,lem:monoEve}: 
   \begin{align*}
       &\paretoCurve{\shortcut{\opg{A}}\seqcomp \shortcut{\opg{B}}}{i}\\
       \succLowerOrder &\best{\Big\{\worst{\big\{\semPlay{\indPlay{\sigma_{\eve}}{\sigma_{\adam}}{i}} \ \big | \ \text{entrance-dependent $\sigma_{\adam}$ on $\opg{A}\seqcomp \opg{B}$} \big\}} \ \Big | \ \text{positional $\sigma_{\eve}$ on $\opg{A}\seqcomp \opg{B}$} \Big\}}.
   \end{align*}
   By~\cref{thm:positionalDeterminacy,lem:lowerOrder}, we can conclude that $\paretoCurve{\opg{A}\ast \opg{B}}{i} = \paretoCurve{\shortcut{\opg{A}}\seqcomp \shortcut{\opg{B}}}{i}$.  

\end{proof}

\end{document}